\documentclass[11pt]{article}

\usepackage{amsthm}
\usepackage{graphicx} 
\usepackage{array} 

\usepackage{amsmath, amssymb, amsfonts, verbatim}
\usepackage{hyphenat,epsfig,subcaption,multirow}
\usepackage{nicefrac}
\usepackage{paralist}

\usepackage[font=footnotesize,labelfont=bf]{caption}

\usepackage[usenames,dvipsnames]{xcolor}
\usepackage[ruled]{algorithm2e}

\DeclareFontFamily{U}{mathx}{\hyphenchar\font45}
\DeclareFontShape{U}{mathx}{m}{n}{
      <5> <6> <7> <8> <9> <10>
      <10.95> <12> <14.4> <17.28> <20.74> <24.88>
      mathx10
      }{}
\DeclareSymbolFont{mathx}{U}{mathx}{m}{n}
\DeclareMathSymbol{\bigtimes}{1}{mathx}{"91}

\usepackage{tcolorbox}
\tcbuselibrary{skins,breakable}
\tcbset{enhanced jigsaw}

\usepackage[normalem]{ulem}
\usepackage[compact]{titlesec}

\definecolor{DarkRed}{rgb}{0.5,0.1,0.1}
\definecolor{DarkBlue}{rgb}{0.1,0.1,0.5}

\usepackage{nameref}
\definecolor{ForestGreen}{rgb}{0.1333,0.5451,0.1333}
\definecolor{Red}{rgb}{0.9,0,0}
\usepackage[linktocpage=true,
	pagebackref=true,colorlinks,
	linkcolor=DarkRed,citecolor=ForestGreen,
	bookmarks,bookmarksopen,bookmarksnumbered]
	{hyperref}
\usepackage[noabbrev,nameinlink]{cleveref}
\crefname{property}{property}{Property}
\creflabelformat{property}{#2(#1)#3}
\crefname{equation}{eq}{Eq}
\creflabelformat{equation}{#2(#1)#3}

\usepackage{bm}
\usepackage{url}
\usepackage{xspace}
\usepackage[mathscr]{euscript}

\usepackage{tikz}
\usetikzlibrary{arrows}
\usetikzlibrary{arrows.meta}
\usetikzlibrary{shapes}
\usetikzlibrary{backgrounds}
\usetikzlibrary{positioning}
\usetikzlibrary{decorations.markings}
\usetikzlibrary{patterns}
\usetikzlibrary{calc}
\usetikzlibrary{fit}
\usetikzlibrary{snakes}
\tikzset{vertex/.style={circle, black, fill=Yellow, line width=1pt, draw, minimum width=8pt, minimum height=8pt, inner sep=0pt}}
	
\usepackage{mdframed}

\usepackage[noend]{algpseudocode}
\makeatletter
\def\BState{\State\hskip-\ALG@thistlm}
\makeatother

\usepackage{cite}
\usepackage{enumitem}

\usepackage[margin=1in]{geometry}

\usepackage{diagbox}
\usepackage{mathabx}

\newtheorem{theorem}{Theorem}
\newtheorem{lemma}{Lemma}[section]
\newtheorem{proposition}[lemma]{Proposition}
\newtheorem{corollary}[lemma]{Corollary}
\newtheorem{claim}[lemma]{Claim}
\newtheorem{fact}[lemma]{Fact}

\newtheorem*{claim*}{Claim}
\newtheorem*{theorem*}{Theorem}
\newtheorem*{proposition*}{Proposition}
\newtheorem*{lemma*}{Lemma}
\newtheorem*{problem*}{Problem}

\crefname{lemma}{Lemma}{Lemmas}
\crefname{claim}{Claim}{Claims}

\newtheorem{mdresult}{Result}
\newenvironment{result}{\begin{mdframed}[backgroundcolor=lightgray!40,topline=false,rightline=false,leftline=false,bottomline=false,innertopmargin=2pt, innerleftmargin=10pt]\begin{mdresult}}{\end{mdresult}\end{mdframed}}

\newtheorem{remark}[lemma]{Remark}
\newtheorem*{remark*}{Remark}

\newtheoremstyle{restate}{}{}{\itshape}{}{\bfseries}{~(restated).}{.5em}{\thmnote{#3}}
\theoremstyle{restate}

\theoremstyle{definition}
\newtheorem{mdalg}{Algorithm}
\newenvironment{Algorithm}{\begin{tbox}\begin{mdalg}}{\end{mdalg}\end{tbox}}

\newtheorem{mddist}{Distribution}
\newenvironment{Distribution}{\begin{tbox}\begin{mddist}}{\end{mddist}\end{tbox}}

\newtheorem*{mdinvariant}{Example}
\newenvironment{example}{\begin{mdframed}[hidealllines=false,innerleftmargin=5pt,backgroundcolor=gray!10,innertopmargin=2pt]\begin{mdinvariant}}{\end{mdinvariant}\end{mdframed}}

\allowdisplaybreaks

\renewcommand{\qed}{\nobreak \ifvmode \relax \else
      \ifdim\lastskip<1.5em \hskip-\lastskip
      \hskip1.5em plus0em minus0.5em \fi \nobreak
      \vrule height0.75em width0.5em depth0.25em\fi}

\newcommand{\misHalf}[1][]{\ifx.#1.{\ensuremath{\mathcal{H}_{\mathsf{MIS}}}}\else{\ensuremath{\mathcal{H}^{(#1)}_{\mathsf{MIS}}}}\fi}
\newcommand{\misHard}[1][]{\ifx.#1.{\ensuremath{\mathcal{D}_{\mathsf{MIS}}}}\else{\ensuremath{\mathcal{D}^{(#1)}_{\mathsf{MIS}}}}\fi}
\newcommand{\apxHard}[1][]{\ifx.#1.{\ensuremath{\mathcal{D}_{\mathsf{MM}}}}\else{\ensuremath{\mathcal{D}^{(#1)}_{\mathsf{MM}}}}\fi}
\newcommand{\pb}[1][]{\ifx.#1.{\mathscr{P}}\else{\mathscr{P}(#1)}\fi}
\newcommand{\fb}[1][]{\ifx.#1.{\mathscr{F}}\else{\mathscr{F}(#1)}\fi}
\newcommand{\sfP}{P}
\newcommand{\sfF}{F}
\newcommand{\halfp}{{\hat{p}}}
\newcommand{\halff}{{\hat{f}}}
\newcommand{\halfn}{{\hat{n}}}
\newcommand{\perm}{\sigma}
\newcommand{\rPerm}{\ensuremath{\rv{\Sigma}}\xspace}
\newcommand{\rO}{\rv{O}}
\newcommand{\rS}{\rv{S}}
\newcommand{\rT}{\rv{T}}
\newcommand{\rW}{\rv{W}}
\newcommand{\const}{{20}}

\newcommand{\tvd}[2]{\ensuremath{\norm{#1 - #2}_{\mathrm{tvd}}}}

\newcommand{\Bracket}[1]{\Big[#1\Big]}
\newcommand{\bracket}[1]{\left[#1\right]}
\newcommand{\paren}[1]{\ensuremath{\left(#1\right)}\xspace}
\newcommand{\card}[1]{\left\vert{#1}\right\vert}

\newcommand{\IN}{\ensuremath{\mathbb{N}}}

\newcommand{\norm}[1]{\ensuremath{\|#1\|}}

\newcommand{\prob}[1]{\Pr\paren{#1}}

\newcommand{\set}[1]{\ensuremath{\left\{ #1 \right\}}}
\newcommand{\poly}{\mbox{\rm poly}}
\newcommand{\polylog}{\mbox{\rm  polylog}}

\DeclareMathOperator*{\Exp}{\ensuremath{{\mathbb{E}}}}
\DeclareMathOperator*{\Prob}{\ensuremath{\textnormal{Pr}}}
\renewcommand{\Pr}{\Prob}

\newcommand{\Ex}{\Exp}

\newenvironment{tbox}{\begin{tcolorbox}[
		enlarge top by=5pt,
		enlarge bottom by=5pt,
		 boxsep=2pt,
                  left=5pt,
                  right=7pt,
                  top=10pt,
                  arc=0pt,
                  boxrule=1pt,toprule=1pt,
                  colback=white
                  ]
	}
{\end{tcolorbox}}

\newcommand{\rv}[1]{\ensuremath{{\mathsf{#1}}}\xspace}
\newcommand{\rA}{\rv{A}}
\newcommand{\rB}{\rv{B}}
\newcommand{\rC}{\rv{C}}
\newcommand{\rD}{\rv{D}}

\newcommand{\rY}{\rv{Y}}
\newcommand{\supp}[1]{\ensuremath{\textnormal{\text{supp}}(#1)}}
\newcommand{\distribution}[1]{\ensuremath{\textnormal{dist}(#1)}\xspace}

\newcommand{\kl}[2]{\ensuremath{\mathbb{D}(#1~||~#2)}}
\newcommand{\II}{\ensuremath{\mathbb{I}}}
\newcommand{\HH}{\ensuremath{\mathbb{H}}}
\newcommand{\mi}[2]{\ensuremath{\def\mione{#1}\def\mitwo{#2}\mireal}}
\newcommand{\mireal}[1][]{
  \ifx\relax#1\relax%
    \II(\mione \,; \mitwo)%
  \else%
    \II(\mione \,; \mitwo\mid #1)%
  \fi
}
\newcommand{\en}[1]{\ensuremath{\HH(#1)}}

\newcommand{\itfacts}[1]{\Cref{fact:it-facts}-(\ref{part:#1})\xspace}

\newcommand{\prot}{\pi}

\newcommand{\rX}{\rv{X}}
\newcommand{\rZ}{\rv{Z}}
\newcommand{\rM}{\rv{M}}

\newcommand{\rG}{\rv{G}}

\title{Rounds vs Communication Tradeoffs \\for Maximal Independent Sets}
\author{Sepehr Assadi\footnote{(\texttt{sepehr.assadi@rutgers.edu}) Department of Computer Science, Rutgers University.  }\and 
Gillat Kol\footnote{(\texttt{gillat.kol@gmail.com}) Department of Computer Science, Princeton University.}  \and 
Zhijun Zhang\footnote{(\texttt{zhijunz@princeton.edu}) Department of Computer Science, Princeton University.}  
}

\date{}

\begin{document}
\maketitle

\pagenumbering{roman}


\begin{abstract}

We consider the problem of finding a maximal independent set (MIS) in the shared blackboard communication model with vertex-partitioned inputs. 
There are $n$ players corresponding to vertices of an undirected graph, and each player sees the edges incident on its vertex -- this way, each edge is known by both its endpoints and is thus \emph{shared} by two players.   
The players communicate in simultaneous rounds by posting their messages on a shared blackboard visible to all players, with the goal of computing an MIS of the graph.
While the MIS problem is well studied in other distributed models, and while shared blackboard is, perhaps, the simplest broadcast model, lower bounds for our problem were only known against one-round protocols.

\medskip

We present a lower bound on the \textbf{round-communication tradeoff} for computing an MIS in this model. Specifically, we show that when $r$ rounds of interaction are allowed, at least one player needs to 
communicate $\Omega(n^{1/\const^{r+1}})$ bits. In particular, with logarithmic bandwidth, finding an MIS requires $\Omega(\log\log{n})$ rounds. This lower bound can be compared with 
the algorithm of Ghaffari, Gouleakis, Konrad, Mitrović, and Rubinfeld [PODC 2018] that solves MIS in $O(\log\log{n})$ rounds but with a logarithmic bandwidth for an \emph{average} player. Additionally, 
our lower bound  further extends to the closely related problem of maximal bipartite matching. 

\medskip

The presence of edge-sharing gives the algorithms in our model a surprising power and numerous algorithmic results exploiting this power are known. 
For a similar reason, proving lower bounds in this model is much more challenging, as this sharing in the players' inputs  prohibits the use of standard number-in-hand communication complexity arguments. 
Thus, to prove our results, we devise a new round elimination framework, which we call \textbf{partial-input embedding}, that may also be useful in future work for proving \emph{round-sensitive} lower bounds in the presence of shared inputs. 

\medskip

Finally, we discuss several implications of our results to multi-round (adaptive) distributed sketching algorithms, broadcast congested clique, and to the welfare maximization problem in two-sided matching markets.

\end{abstract}

\clearpage

\setcounter{tocdepth}{3}
\tableofcontents

\clearpage

\pagenumbering{arabic}
\setcounter{page}{1}


\section{Introduction}\label{sec:intro}

Consider the following communication model: there are $n$ players corresponding to vertices of an undirected graph $G=(V,E)$ and each player only sees the edges incident on its vertex -- this way, each 
edge of the graph is \emph{shared} by the two players at its endpoints. The goal of the players 
is to solve some fixed problem on $G$, for instance, finding a spanning forest of $G$. To do so, the players communicate in synchronous rounds wherein all parties simultaneously write a message on a shared blackboard visible to all.
The messages communicated by the players are only functions of their own inputs and the content of the blackboard. When the protocol concludes, an additional party, called the {\em referee}, computes the output of the protocol as a function of the blackboard 
content. We are interested in the tradeoff between the number of rounds of the protocol and the \emph{per-player} communication, defined as the worst-case length of any message sent by any player in any round. 

In the communication complexity terminology, this model is referred to as the multi-party communication model with {\em shared blackboard} and {\em vertex-partitioned inputs}. However, it has also been studied by different communities under different names, such as 
{\em broadcast congested clique}~\cite{DruckerKO13,BeckerMRT18,JurdzinskiL018a,Jurdzinski018b}, or  {\em (adaptive) distributed sketching}~\cite{AhnGM12a,AhnGM12b,AssadiKO20,FiltserKN21}. At this point, there is quite a  large body of 
algorithmic results in this model~\cite{AhnGM12a,AhnGM12a,AhnGM12b,AhnGM13,KapralovLMMS14,GuhaMT15,McGregorTVV15,AssadiCK19,AssadiKM22,LattanziMSV11,KapralovW14,FiltserKN21,AhnCGMW15} (see~\Cref{sec:models} for more
details). The source of power behind these results is a crucial aspect of this model: \emph{edge-sharing}, or in other words, the fact that each edge of the graph is seen by both its endpoints\footnote{The interested reader is referred to~\cite{AhnGM12a} to see this in a surprising algorithm that solves graph connectivity using only a single round and $O(\log^3{n})$ communication bits per player.}. This sharing in the players' inputs makes this model an ``intermediate'' model lying between the {\em number-in-hand} model (with no input sharing) and the notorious {\em number-on-forehead} model (with arbitrary input sharing). As a result, lower bounds are more scarce in this model~\cite{BeckerMNRST11,DruckerKO13,BeckerMRT14,BeckerMRT18,JurdzinskiL018,NelsonY19,Yu21,AssadiKO20}.  

We study the \emph{maximal independent set (MIS)} problem in this model. While MIS is one of the most studied problems in other distributed models (see, e.g.,~\cite{Luby85,Linial87,KuhnMW16,Ghaffari16,BalliuBHORS19}),
and while shared blackboard is, perhaps, the simplest broadcast model, not much is known about MIS in this model. We do note that Luby's celebrated MIS algorithm~\cite{Luby85} implies an $O(\log{n})$-round $O(1)$-per-player communication
algorithm in this model. Ghaffari, Gouleakis, Konrad, Mitrovic, and Rubinfeld~\cite{GhaffariGKMR18} gave an algorithm that runs in $O(\log\log{n})$ rounds, but only bounds the communication of an \emph{average} player by $O(\log{n})$. I.e., 
the total communication by all players in a round is $O(n\log{n})$, but some players may need to communicate $\omega(\log{n})$ bits\footnote{This algorithm is designed for the (unicast) congested clique model, but given its connection to 
the distributed sketching/dynamic streaming algorithm of~\cite{AhnCGMW15}---that solves MIS as a subroutine in correlation clustering---it can be directly implemented in our model with the mentioned bounds.}. Moreover,
Assadi, Kol, and Oshman proved that any \emph{one}-round protocol requires almost $(n^{1/2})$ per-player communication. This state-of-affairs raises the following question: 
\begin{quote}
	\emph{What is the complexity of MIS in the shared blackboard model with vertex-partitioned inputs? In particular, what are the possible round-communication tradeoffs in this model? } 
\end{quote}
We make progress on this fundamental open question by presenting a new lower bound on the round-communication tradeoff for the MIS problem. 
The key contribution of our work is a new technique for proving multi-round lower bounds, even in the presence of edge-sharing. 
 This also allows us to prove a similar lower bound for another fundamental problem, namely, the maximal bipartite matching problem. 

Our work can be viewed as a direct continuation of two lines of work: the first line of work is on \emph{number-in-hand} multi-round communication complexity, where we follow up on the result of Alon, Nisan, Raz, and Weinstein~\cite{AlonNRW15}. They give lower bounds for the bipartite maximal matching problem, where only parties on one side of the partition are allowed to communicate. The second line of work is the aforementioned lower bound of Assadi, Kol, and Oshman~\cite{AssadiKO20}, which works in our model, but only considers one-round protocols. In the following, we elaborate more on our results, techniques, 
and their connections to other settings.


\subsection{Our Contributions}


Our main result is a multi-round lower bound for computing MIS in the shared blackboard model. 

\begin{result}\label{res:mis}
	Any $r$-round multi-party protocol (deterministic or randomized) in the shared blackboard model for finding a maximal independent set on $n$-vertex graphs
	 requires $\Omega(n^{1/\const^{r+1}})$ bits of communication per player. In particular, $\Omega(\log\log{n})$ rounds are needed for 
	protocols with $\poly\!\log\!{(n)}$ per-player communication. 
\end{result}

Previously, the only known lower bound for MIS in our model was the (almost) $\Omega(n^{1/2})$-communication lower bound of~\cite{AssadiKO20} for one-round protocols. Indeed, to the best of our knowledge, 
there has been no prior communication lower bound in this model for any natural problem that is \emph{sensitive} to the number of rounds (the lower bounds were either for one-round protocols, e.g.,~\cite{NelsonY19,AssadiKO20,Yu21}, or 
arbitrary number of rounds, e.g.,~\cite{DruckerKO13,BeckerMRT18}\footnote{Specifically, the latter ones bound the \emph{total} communication needed to solve the problem and use this to get a lower bound on the number of rounds \emph{times} communication per round. Such lower bounds cannot capture more nuanced round-communication tradeoffs (e.g., like the ones exhibited by~\cite{Luby85} or~\cite{GhaffariGKMR18} for MIS).}).

The tradeoff achieved in~\Cref{res:mis} asymptotically matches the aforementioned $O(\log\log{n})$-round algorithm of~\cite{GhaffariGKMR18} for finding MIS, except that, as mentioned before, the protocol of~\cite{GhaffariGKMR18} 
only bounds the communication of an \emph{average} player by $O(\log{n})$ bits and a few players need to communicate way more than $\polylog{(n)}$ bits. Thus, the two results do not directly match. 
It remains an interesting open question to either improve the guarantee of the algorithm of~\cite{GhaffariGKMR18} to per-player communication bound or improve our lower bound to average-case communication. 

\smallskip

Our techniques in establishing~\Cref{res:mis} are quite general and, as a corollary to our proof, also allow us to prove a lower bound for another fundamental problem, namely, maximal matching. 

\begin{result}\label{res:mm}
	Any $r$-round multi-party protocol (deterministic or randomized) in the shared blackboard model for finding a maximal matching or any constant factor approximation to maximum matching on $n$-vertex (bipartite) graphs
	 requires $\Omega(n^{1/\const^{r+1}})$ bits of communication per player. As such, $\Omega(\log\log{n})$ rounds are needed for 
	protocols with $\poly\!\log\!{(n)}$ per-player communication. 
\end{result}

As in the case of MIS, the only known lower bound prior to our work was the one-round lower bound of~\cite{AssadiKO20}. However, for the \emph{number-in-hand} variant of our communication model, wherein each edge of the 
graph is only seen by one of its endpoints, a series of papers~\cite{DobzinskiNO14,AlonNRW15,BravermanO17} proved a nearly-logarithmic round lower bound for the matching problem (we elaborate on this line of work later). 
Yet, the number-in-hand model is algorithmically much weaker than the edge-sharing model studied in our paper; for instance, the lower bound of~\cite{BravermanO17} also holds for finding a spanning forest of the input in that model, while 
finding spanning forests in our model can be done with $O(\log^3{n})$ communication in just one round~\cite{AhnGM12a}. We refer the reader to~\cite{AssadiKO20} for discussions on the inherit difference of number-in-hand model
and our model that allows for edge-sharing and thus is ``one step closer'' to the notorious number-on-forehead model. 

\paragraph{Our techniques.} We shall go over our techniques in detail in the streamlined overview of our
approach in~\Cref{sec:overview}. For now, we only mention the high level bits of our techniques. 

Our techniques unify and generalize the lower bounds of~\cite{AssadiKO20} for one-round protocols in our model, as well as the lower bounds of~\cite{AlonNRW15} for multi-round protocols in the number-in-hand model. 
To this end, we need several substantially new ideas\footnote{Braverman and Oshman~\cite{BravermanO17} gave stronger lower bounds than~\cite{AlonNRW15}, that work for nearly logarithmic number of rounds. However, their techniques seem ``too tailored'' to the number-in-hand model and approximate matchings, and thus are \emph{not} suitable for us (given the algorithm of~\cite{GhaffariGKMR18} for MIS, which, even though not exactly in our model, seem quite close, it is not clear if one can get a logarithmic lower bound in our model).}. The main novelty of our work is in developing a new \emph{round elimination} argument 
that is tailored to our edge-sharing model. Similar to standard round elimination arguments, say, the one in~\cite{AlonNRW15}, our approach is also based on \emph{simulating} an $r$-round protocol on ``large'' instances in only $(r-1)$ rounds
for smaller ``embedded'' instances (with fewer players and smaller inputs). Prior work perform such a simulation by generating an \emph{input} for the ``missing'' players of the large $r$-round instance
with \emph{low correlation} with the actual embedded $(r-1)$-round hard instance. As we argue, such an approach is doomed to fail for our model with its edge-sharing aspects. Instead, we introduce a \emph{partial-input embedding} argument 
that implements this simulation via generating only the \emph{messages} of the missing players. We then use information-theoretic tools to track the gradual increase in the \emph{correlation} of these messages with the embedded hard instance
throughout the \emph{entire} simulation (not only the first round which is sufficient for ``input-sampling'' protocols of prior work). 

\subsection{Further Implications of Our Results to Related Models}\label{sec:models}

We conclude this section by listing further implications of our results to other well-studied settings. 

\paragraph{Broadcast congested clique.} The communication model studied in our paper is equivalent to the broadcast congested clique model studied in various prior work, e.g., in~\cite{DruckerKO13,BeckerMRT18,JurdzinskiL018a,Jurdzinski018b}. 
Specifically, our~\Cref{res:mis} and~\Cref{res:mm} imply $\Omega(\log\log{n})$ round lower bounds for both MIS and maximal matching on any broadcast congested clique algorithm with $\polylog{(n)}$ bandwidth. 
Incidentally, in the stronger \emph{unicast} congested clique model, $O(\log\log{n})$-round algorithms are known for both MIS~\cite{GhaffariGKMR18} and maximal matching~\cite{BehnezhadHH19}. We note that, as shown in~\cite{DruckerKO13}, 
proving lower bounds in the unicast model implies strong circuit lower bounds and thus is beyond the reach of current techniques. 

\paragraph{Distributed sketching.} Our model is also equivalent to the distributed sketching model that was initiated in the breakthrough work of~\cite{AhnGM12a}. 
Starting from the connectivity sketch of~\cite{AhnGM12a}, there has been tremendous progress on efficient distributed sketching algorithms for various other problems in \emph{one} round, e.g.,~cut sparsifiers~\cite{AhnGM12b}, spectral sparsifiers~\cite{AhnGM13,KapralovLMMS14}, vertex connectivity~\cite{GuhaMT15}, densest subgraph~\cite{McGregorTVV15}, $(\Delta+1)$-coloring~\cite{AssadiCK19}, $\Delta$-coloring~\cite{AssadiKM22}, and in \emph{multiple} rounds, e.g., minimum spanning trees~\cite{AhnGM12a}, matchings~\cite{LattanziMSV11,AhnGM12a}, spanners~\cite{KapralovW14,FiltserKN21}, and MIS and correlation clustering~\cite{AhnCGMW15}. Given the strength of this model, proving lower bounds in this model has been a highly challenging task (see, e.g.~\cite{AssadiKO20,FiltserKN21}), and only a handful of lower bounds are known including $\Omega(\log^3{n})$ bits 
for connectivity~\cite{NelsonY19,Yu21} and $\Omega(n^{1/2})$ bits for MIS and maximal matching~\cite{AssadiKO20} for \emph{one}-round sketches. Our results contribute to this line of work by providing the first \emph{round-sensitive} lower bounds
in this model, and our techniques can be of independent interest here as well. 


\paragraph{Dynamic streaming algorithms.} One key motivation of~\cite{AhnGM12a} in introducing graph sketching was their application to \emph{dynamic} (semi-)streaming algorithms that can process streams of insertions and deletions of edges with $O(n \cdot \poly\log{(n)})$ memory 
(\emph{all} sketches mentioned above also imply dynamic streaming algorithms). Multi-round sketching protocols, similar to the ones  in our model, then correspond to multi-pass streaming algorithms. 
Currently, the best known multi-pass dynamic semi-streaming algorithms for MIS and maximal matching require $O(\log\log{n})$ passes~\cite{AhnCGMW15} and $O(\log{n})$ passes~\cite{LattanziMSV11,AhnGM12a}, respectively. On the lower bound front
however, only single-pass lower bounds are known for either problem~\cite{AssadiKLY16,AssadiCK19,CormodeDK19,DarkK20} (there has been recent progress on multi-pass lower bounds for computing \emph{exact} maximum matchings ~\cite{GuruswamiO13,AssadiR20,ChenKPS0Y21} in logarithmic passes or even $(1+o(1))$-approximation in two passes~\cite{Assadi22} but they do \emph{not} apply to maximal matching in any way). While our results do \emph{not} imply streaming lower bounds, they do rule out certain popular techniques of 
vertex-partitioned graph sketching for obtaining $o(\log\log{n})$-pass algorithms for either problem. Thus, they can form a starting point for proving multi-pass lower bounds for all dynamic streaming algorithms as well. 

\paragraph{Welfare maximization and interaction.} A beautiful line of work initiated by~\cite{DobzinskiNO14} and followed up in~\cite{AlonNRW15,BravermanO17,Assadi17ca,Nisan21}, studies the role played by the \emph{interaction} of participating agents in the efficiency of markets. One formalization, corresponding to unit-demand agents in a matching market, is as follows: we have $n$ agents who are interested in getting any one of their private subset of $n$ items; the goal is to allocate these items in a way that maximizes the \emph{welfare}, defined as the number of agents who receive an item of their liking. The market proceeds in rounds wherein the agents communicate $\polylog{(n)}$-bit messages about their desired items. How many rounds of interaction are needed to maximize the welfare to within a constant factor? 

This problem can be seen as approximating matchings on the bipartite graph consisting of agents on one side that have edges to their preferred items on the other side. The model of communication is also identical to the one in our paper with the crucial 
difference that only vertices on one side of the bipartition, namely, the agents, are communicating. In this model,~\cite{DobzinskiNO14} gave an $O(\log{n})$-round algorithm and ruled out one-round algorithms. \cite{AlonNRW15}
improved the lower bound to $\Omega(\log\log{n})$ rounds and subsequently~\cite{BravermanO17} obtained a nearly tight $\Omega(\frac{\log{n}}{\log\log{n}})$ lower bound (similar lower bounds are also
obtained for the more general setting of combinatorial auctions in~\cite{Assadi17ca}). 

All these results are restricted to one-sided markets. Our~\Cref{res:mm} generalizes (some of) these results to \emph{two-sided} matching markets~\cite{roth1992two}, wherein \emph{both} sides of the market consist of 
communicative agents that know in advance if they make a good match. A canonical example of two-sided matching markets is college admissions and the celebrated Gale-Shapley algorithm for stable marriage~\cite{gale1962college}. Another example, perhaps more closely related to the setting of our paper, is assigning users to servers in a large distributed Internet service~\cite{MaggsS15}. Our~\Cref{res:mm} suggests that even when both sides of the market are able to communicate with a limited bandwidth,  at least a modest amount of interaction is necessary for maximizing welfare (approximately). 

\clearpage


\section{Preliminaries}\label{sec:prelim}

\paragraph{Notation.} For an integer $t \in \mathbb{N}$, we write $[t]$ as a shorthand for the set $\set{1,\ldots,t}$.
Let $h : A \to B$ be an arbitrary function for two sets $A,B$.
For any subset $Z \subseteq A$, we use $h(Z) = \set{h(z) \mid z \in Z}$. 
For a tuple $X = (X_1,\ldots,X_t)$ and integer $i \in [t]$, we define $X_{<i} = (X_1,\ldots,X_{i-1})$ (we also define $X_{-i}$ and $X_{\leq i}$   analogously). 
For a graph $G = (V,E)$ and a permutation $\perm$ over $V$, we denote by $\perm(G)$ the graph on the same vertex set in which $\perm(u)$ and $\perm(v)$ are connected if and only if $(u,v) \in E$.

When there is room for confusion, we use sans-serif letters for random
variables (e.g. $\rA$) and the same normal letters for their realizations (e.g. $A$). For random variables $\rA,\rB$, we use $\supp{\rA}$ as the support of $\rA$, $\en{\rA}$ as the \emph{Shannon entropy}, $\mi{\rA}{\rB}$ as the \emph{mutual information}, 
$\kl{\rA}{\rB}$ as the \emph{KL-divergence}, and $\tvd{\rA}{\rB}$ as the \emph{total variation distance}.  Necessary background on
information theory, including the definitions and basic tools, is provided in \Cref{sec:info}. 

\subsection{Multi-Party Shared Blackboard Model with Vertex-Partitioned Inputs}

We work in the multi-party shared blackboard model with vertex-partitioned inputs, also known as the broadcast congested clique model in the literature.
The communication model is defined formally as follows.
Consider a simple graph $G = (V,E)$ with one player assigned to each of the $n = |V|$ vertices. For convenience, we identify a vertex with its associated player in the rest of this paper and use the two terms interchangeably.
There is a \emph{shared blackboard}, initially empty, that is readable and writable by all players.
The player associated to a vertex $v \in V$ is presented as input with $n$, a unique ID of $v$ in the range $[n]$, and IDs of all of $v$'s neighbors $N_G(v) = \set{u \in V \mid (v,u) \in E}$.
Thus, each edge $(u,v) \in E$ is \emph{shared} by both players $u$ and $v$.

Communication proceeds in $r \in \IN$ \emph{synchronous} rounds.
For each round $t \in [r]$, the players compute their messages based on their initial input as well as the current content of the blackboard, and post them to the blackboard \emph{simultaneously}.
In a randomized protocol, the players may also use both public and private randomness.
After the last round, the final content of the blackboard constitutes the \emph{transcript}, denoted by $\Pi$, of the protocol.
Then, a \emph{referee} computes the output of the protocol depending  on $\Pi$ (and possibly public randomness of all players and its own private randomness).
The \emph{bandwidth} of a protocol is defined to be the \emph{maximum} number of bits ever communicated by any player in any round.

We are interested in round-communication tradeoff of the following problems:

\paragraph{Maximal Independent Set.}
We say a protocol computes a {maximal independent set} (MIS) with error probability $\delta \in [0,1]$ if the output of the referee is a valid MIS of $G$ with probability at least $1-\delta$ over the randomness of the protocol.
The protocol may err by outputting a subset of vertices which is not independent or not maximal.

\paragraph{Approximate Matching.}
We say a protocol computes an $\alpha$-\emph{approximate matching} ($\alpha \ge 1$) if the output $\Gamma(\Pi)$ of the referee: $(1)$ is always a set of \emph{disjoint} pairs of vertices; and $(2)$ satisfies $\Exp\card{\Gamma(\Pi) \cap E} \ge \mu(G) / \alpha$, where $\mu(G)$ is the size of the maximum matching of $G$ and expectation is taken over the randomness of the protocol.
This definition allows the referee to output \emph{non-existing} edges as long as they are disjoint but only the correct ones in $E$ are counted.
This is a less restrictive error-model than requiring the algorithm to output a valid matching with certain probability 
and our lower bound holds even in this less restrictive setting; see also~\cite{AlonNRW15}. 


\clearpage


\newcommand{\Ins}{{I}}
\newcommand{\Jns}{{J}}

\newcommand{\rIns}{\rv{\Ins}}
\newcommand{\rJns}{\rv{\Jns}}

\newcommand{\Insstar}{\Ins^{\star}}

\section{Technical Overview}\label{sec:overview}
As our proof is quite dense and technical
and involves various information theoretic maneuvers that are daunting to parse, we use this section to unpack our main ideas and
give a streamlined overview of our approach. We emphasize that
this section oversimplifies many details and the discussions will be informal for the sake of intuition.

The starting point of our approach is a lower bound of~\cite{AlonNRW15} for approximate matchings in 
the \emph{number-in-hand} multi-party communication model. We first give a detailed discussion of this result as our techniques need to inevitably subsume this work (since our result 
implies theirs as well). We then discuss the challenges of extending this result to our model that allows for edge-sharing and present a technical overview of our work. 
We stick with approximate matchings in this overview as it is easier to work with and to compare with~\cite{AlonNRW15}. 

\subsection{A Detailed Overview of~\cite{AlonNRW15}}\label{sec:anrw15} 

\noindent
\!\!\cite{AlonNRW15} considers the same communication setting as ours on bipartite graphs $G=(L \sqcup R,E)$ with the key difference that the players are only associated with vertices in $L$, 
and thus each edge is seen by only a single player. 
They prove that any protocol that uses $\polylog{(n)}$ communication per player and computes an $O(1)$-approximate matching 
requires $\Omega(\log\log{n})$ rounds in this model. 

The proof in~\cite{AlonNRW15} is via \textbf{round elimination}: to lower bound $\poly\!\log{(n)}$-communication $r$-round protocols $\prot_r$, they start with $p_r \approx n^{4/5}$ \emph{independent} $(r-1)$-round ``hard'' instances ${\Ins_1,\ldots,\Ins_{p_r}}$, called \emph{principal} instances. These instances are supported on disjoint sets of $\approx n^{1/5}$ vertices each, and are then ``embedded''  in a single graph $G$ to form an $r$-round instance $\Ins$. This instance is such that the \emph{first} message of $\prot_r$ cannot reveal much information about  principal instances and thus $\prot_r$ cannot 
solve them in its remaining $r-1$ rounds given their (inductive) hardness. 

\begin{figure}[h!]
	\centering

\begin{tikzpicture}

\tikzset{layer/.style={rectangle, rounded corners=5pt, draw, black, line width=1pt,  fill=black!10, inner sep=4pt}}
\tikzset{vertex/.style={circle, ForestGreen, fill=white, line width=2pt, draw, minimum width=8pt, minimum height=8pt, inner sep=0pt}}
\tikzset{choose/.style={rectangle, line width=1pt, rounded corners = 2pt, draw, minimum width=40pt, minimum height=16pt, fill=black!10}}
		
\node (R) {$R$:}; 
\node (L) [below=40pt of R]{$L$:}; 		
\node[choose](J1)[right=1pt of R]{};
\node[choose](J2)[right=6pt of J1]{};
\node[choose](J3)[right=6pt of J2]{};

\node[choose, fill=blue!20](x1) [right=30pt of J3]{};
\node[choose, fill=blue!20](y1) [below=40pt of x1]{};

\foreach \i in {2,...,6}
{
	   \pgfmathtruncatemacro{\ip}{\i-1};
	   
	   \node[choose, fill=blue!20] (x\i) [right=6pt of x\ip]{};
	   \node[choose, fill=blue!20] (y\i) [right=6pt of y\ip]{};

}

\foreach \i in {1,...,6}
{
	   \draw[line width=2pt, blue] (x\i) -- node [left] {$\Ins_{\i}$} (y\i);
	   
	   \foreach \j in {1,2,3}
	   {
	   	 \draw[line width=1pt, gray, opacity=0.25] (y\i) -- (J\j);
	   }
}

\foreach \j in {1,2,3}
{
   	 \draw[line width=1pt, black] (y1) -- node [pos=0.7, left] {$\Jns_{1,\j}$}(J\j);
}


\end{tikzpicture}
	\caption{An illustration of the lower bound instances of~\cite{AlonNRW15} with parameters $f_r = 3$ and $p_r = 6$. The top right vertices (blue) are used in principal instances, while top left vertices (gray) are fooling instances. The heavy (blue) edges are from principal instances and the light (gray) edges are from fooling instances -- to avoid clutter, only the edges in fooling instances of the first principal instance are drawn (solid black edges). 
To find a large matching in this graph, one needs to find sufficiently large matchings in many of the principal instances. 
}\label{fig:anrw15}
\end{figure}
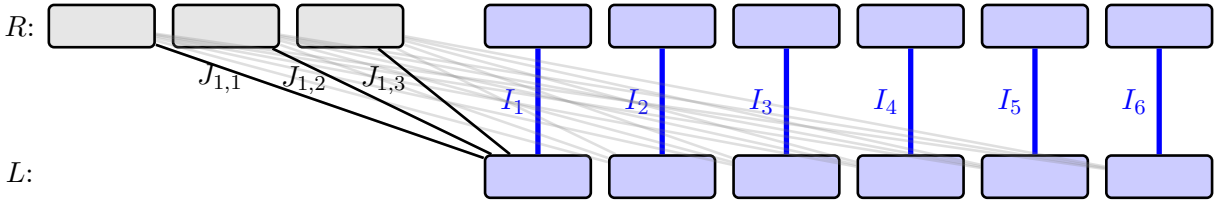
\vspace{-5pt}

To limit the information revealed by $\prot_r$ about principal instances,~\cite{AlonNRW15} further ``packs'' the graph, for every principal instance $i \in [p_r]$,  
with $f_r \approx n^{2/5}$ \emph{fooling} instances $\Jns_{i,*} := \Jns_{i,1},\ldots,\Jns_{i,f_r}$. This packing ensures that: $(1)$  these fooling instances are supported on a small set of vertices on the $R$-side of the bipartition and so $\prot_r$ still has to solve \emph{most of} the underlying principal instances in order
to solve $\Ins$; and $(2)$ each player in $\Ins$
``plays'' in $f_r+1$ instances, consisting of only one principal instance, while being \emph{oblivious} to which instance is the principal one. An ingenious idea in~\cite{AlonNRW15} is that 
these fooling instances need \emph{not} actually be hard $(r-1)$-round instances! Instead, they form a {product} distribution where  
for each vertex $v \in L$, only the {marginal} distribution of $v$ is the same under fooling and principal instances. This ensures that 
in the first round (and only in this round), $v$ cannot distinguish between principal and fooling instances. 

\vspace{-1pt}
\paragraph{Round elimination embedding.} We can now discuss how \cite{AlonNRW15} \emph{eliminates} the first round of  $\prot_r$ and obtains 
an $(r-1)$-round protocol $\sigma$ for solving an $(r-1)$-round hard instance $\Insstar$. 

\vspace{-5pt}
\begin{tbox}
\vspace{-5pt}
\begin{center}
\underline{Embedding argument of~\cite{AlonNRW15}:}
\end{center}
\begin{enumerate}[label=$(\roman*)$,leftmargin=20pt]
\item The players in $\sigma$  sample the first message $M^{(1)}$ of $\prot_r$ using \emph{public} randomness. 
\item Then, they will sample an index $i \in [p_r]$ uniformly and let $\Ins_i = \Insstar$ in the instance $\Ins$. 
\item Next, they sample 
$\Jns_{i,1},\ldots,\Jns_{i,f_r}$ conditioned on $M^{(1)}$ and $\Ins_i = \Insstar$ using \emph{private} randomness. This is a non-trivial sampling process which, on a high level, is doable only because fooling instances are {product} distributions (with only the marginals matching principal ones). 

More specifically, each player $v$ \emph{independently} sample its own input $\Jns_{i,*}(v)$ in all the fooling instances, conditioned on only its actual input $\Ins_i(v)$ in its principal instance $\Ins_i$, and $M^{(1)}$.

\item Finally, the players of $\sigma$ sample the {remaining} $p_r-1$ principal instances $\Ins_{-i}$ and $(p_r - 1) \cdot f_r$ fooling instances $\Jns_{-i,*}$  conditioned on $M^{(1)}$ to have a complete instance $\Ins$. 
\end{enumerate}
\end{tbox}
\vspace{-5pt}
At this point, the players in $\sigma$ already have the first message $M^{(1)}$ of $\prot_r$ as well as inputs of all underlying instances without any communication. So, they can continue running $\prot_r$ from its second round, by each player of
$\sigma$ on $\Insstar$ communicating the messages of corresponding player of $\prot_r$ in $\Ins_i$, and simulating messages of $\prot_r$ for players outside $\Ins_i$ with no  communication. 
As $\prot_r$ will also need to solve $\Ins_i$ for a random $i \in [p_r]$, this gives a $(r-1)$-round protocol $\sigma$ for $\Insstar=\Ins_i$. 

At a high level, the correctness of this approach can be argued as follows: 
\vspace{-5pt}
\begin{itemize}[leftmargin=15pt]
	\item The \emph{right} distribution of all underlying variables for $\prot_r$ can be expressed as (by chain rule): 
	\begin{align}
		\rM^{(1)} \times (\rIns_i \mid \rM^{(1)}) \times (\rJns_{i,*} \mid \rIns_i , \rM^{(1)}) \times (\rIns_{-i},\rJns_{-i,*} \mid \rJns_{i,*} , \rIns_i , \rM^{(1)}). \label{eq:anrw-right}
	\end{align}
	\item The distribution sampled from in the protocol $\sigma$ on the other hand is: 
		\begin{align}
		\underbrace{\rM^{(1)}}_{\text{publicly}} \times \underbrace{\rIns_i}_{\text{input}} \times ({\underbrace{{\bigtimes}_{\!v} \rJns_{i,*}(v) \mid \rIns_i(v) , \rM^{(1)}}_{\text{privately}}}) \times \underbrace{(\rIns_{-i},\rJns_{-i,*} \mid \rM^{(1)})}_{\text{publicly}}. \label{eq:anrw-alg}
	\end{align}
\end{itemize}
\vspace{-5pt}

Let us show that these distributions are $o(1)$-close  in total variation distance, which implies that $\prot_r$ also works (almost) as good on sampled instances (see~\Cref{fact:tvd-small}), giving us the desired $(r-1)$-round protocol $\sigma$ for $\Insstar$. Here, the first terms are the same. For the second terms,
\begin{align}
	\tvd{\rIns_i}{(\rIns_i \mid \rM^{(1)})}^2 \leq	 \mi{\rIns_i}{\rM^{(1)}} \leq \frac{1}{f_r+1} \cdot \mi{\rJns_{i,*},\rIns_i}{\rM_i^{(1)}} \leq o(1). \label{eq:anrw-fooled}
\end{align}
In~\Cref{eq:anrw-fooled}, the first inequality is standard (see~\Cref{fact:kl-info} and \Cref{fact:pinskers}). The second 
inequality uses the fact that the players in $\Ins_i$ in $\prot_r$ are oblivious to origins of their edges in $\Ins_i$ vs. $\Jns_{i,*} = \Jns_{i,1},\ldots,\Jns_{i,f_r}$ (by the marginal indistinguishability of these instances); thus, the information revealed by their messages $M_i^{(1)}$ is ``spread'' over these instances; also, other players of $\prot_r$ cannot reveal any information about these instances as they do not see them. 
The final inequality holds because the messages communicated by $\approx n^{1/5}$ players in $\Ins_i$ have collective size much smaller than $f_r \approx n^{2/5}$. 

Finally, the
third and fourth terms in~\Cref{eq:anrw-right,eq:anrw-alg} also have the same distributions in both cases, which at a high level, follows from the rectangle property of communication protocols: for instance, since $\Jns_{i,*}(u)$ and $\Jns_{i,*}(v)$ were independent originally, they remain independent even after conditioning on $M^{(1)}$ -- this is sufficient to show the equivalence of corresponding distributions. This concludes the closeness of these distributions and our overview 
 of the work of~\cite{AlonNRW15}. 
 
 \subsection{Our Approach and New Ideas}\label{sec:approach} 
 
 The very first obvious challenge in using construction of~\cite{AlonNRW15} in our model is that it can be easily solved in just a single round once \emph{both} sides of the bipartite graph can speak (the maximum matching of instances created
 is incident on vertices with degree one in $R$ who can just communicate their edge directly on the blackboard). This brings us to the first and most obvious of our ideas. 
 
\subsubsection{Idea One: Symmetrizing the Input Distribution}\label{sec:idea1}
 The first step is to symmetrize the input distribution in~\cite{AlonNRW15}. Basically, to create a hard $r$-round instance, 
we again start with $(r-1)$-round hard principal instances $\Ins_1,\ldots,\Ins_{p_r}$. We then also add $f_r$ sets of vertices $\fb_1,\ldots,\fb_{f_r}$ called the \emph{fooling blocks} and use 
vertices on \emph{both} sides of each principal instance $\Ins_i$, called \emph{principal block} $\pb_i$, and the fooling blocks to form fooling instances $\Jns_{i,1},\ldots,\Jns_{i,f_r}$ -- as before, these fooling instances
are \emph{not} hard $(r-1)$-round distributions, but only that the input of principal blocks match the ``right'' distribution \emph{marginally}. 

\smallskip
 \begin{figure}[h!]
	\centering

\begin{tikzpicture}

\tikzset{layer/.style={rectangle, rounded corners=5pt, draw, black, line width=1pt,  fill=black!10, inner sep=4pt}}
\tikzset{vertex/.style={circle, ForestGreen, fill=white, line width=2pt, draw, minimum width=8pt, minimum height=8pt, inner sep=0pt}}
\tikzset{choose/.style={rectangle, line width=1pt, rounded corners = 2pt, draw, minimum width=40pt, minimum height=16pt, fill=black!10}}
			
\node[choose](J1){$\fb_1$};
\node[choose](J2)[right=6pt of J1]{$\fb_2$};
\node[choose](J3)[right=6pt of J2]{$\fb_3$};

\node[choose, fill=blue!20](x1) [above right=10pt and 40pt of J3]{$\pb_1$};
\node[choose, fill=blue!20](y1) [below=40pt of x1]{$\pb_1$};

\foreach \i in {2,...,6}
{
	   \pgfmathtruncatemacro{\ip}{\i-1};
	   
	   \node[choose, fill=blue!20] (x\i) [right=6pt of x\ip]{$\pb_{\i}$};
	   \node[choose, fill=blue!20] (y\i) [right=6pt of y\ip]{$\pb_{\i}$};
	
}

\foreach \i in {1,...,6}
{
	   \draw[line width=2pt, blue] (x\i) -- node [left] {$\Ins_{\i}$} (y\i);
	   
	   \foreach \j in {1,2,3}
	   {
	   	 \draw[line width=1pt, gray, opacity=0.25] 
		 (x\i) -- (J\j)
		 (y\i) -- (J\j);
	   }
}

\foreach \j in {1,2,3}
{
   	 \draw[line width=1pt, black] 
	 (y1) -- node [pos=0.65] {$\Jns_{1,\j}$}(J\j.south east)
	 (x1) -- node [pos=0.65] {$\Jns_{1,\j}$}(J\j.north east);
}

\end{tikzpicture}
	\caption{An illustration of our lower bound instances with parameters $f_r = 3$ and $p_r = 6$. The top and bottom vertices (blue) are principal blocks, while middle left vertices (gray) are fooling blocks. 
	The heavy (blue) edges are from principal instances and the light (gray) edges are from fooling instances -- to avoid clutter, only the edges in fooling instances of the first principal instance are drawn (solid black edges). 
	Note that fooling blocks participate only in fooling instances while principal blocks participate both in principal and fooling instances. 
}\label{fig:ours}
\end{figure}
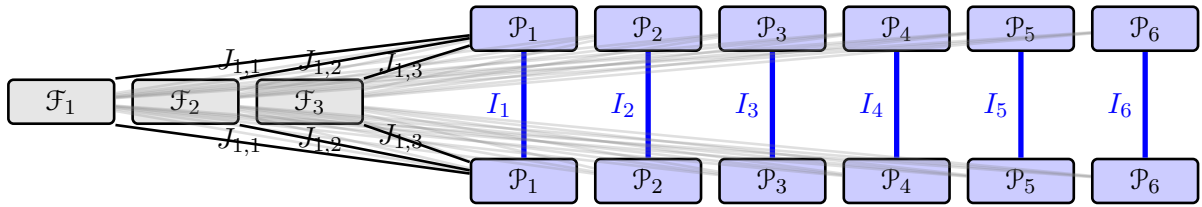

This step of symmetrizing the input distribution is a straightforward extension of~\cite{AlonNRW15}, and we claim no novelty in this part. The interesting part is how to analyze this distribution in our model in light of the following key differences from~\cite{AlonNRW15}: (1) in addition to {principal} blocks, vertices in $\fb_1,\ldots,\fb_{f_r}$ can now also communicate; and (2) there is an edge-sharing aspect in our model; in particular,
sharing of edges between fooling blocks and principal blocks allows fooling blocks to communicate even about edges directly inside principal instances (!), and yet fooling blocks themselves are not even fooled anymore in the distribution. We discuss our approach for handling these parts in the following three subsections. 

\subsubsection{Idea Two: Bounding Revealed Information on Average}\label{sec:idea2}
 Our goal as before is to do a round elimination argument and embed an $(r-1)$-round instance inside an $r$-round one. 
 Our embedding argument in the \emph{first} round is going to be the same as that of~\cite{AlonNRW15}, except that we also sample the first message $M^{(1)}_F$ of fooling blocks  using public randomness  (there are \emph{no} such players in~\cite{AlonNRW15}). 
 We will then  have all the messages of round one, namely, $M^{(1)} = (M^{(1)}_P,M^{(1)}_F)$, as well as edges {incident} on the principal block $\pb_i$, namely, $\Ins_i,\Jns_{i,*}$, inside $\Ins$ without having done any communication. 

Specifically, we design a protocol $\sigma$ that given an $(r-1)$-round instance $\Insstar$, creates an $r$-round instance $\Ins$ and uses a $\polylog{(n)}$-communication $r$-round protocol $\prot_r$ on $\Ins$ to solve $\Insstar$ as follows.
 
 \begin{tbox}
\vspace{-5pt}
\begin{center}
\underline{Our embedding argument -- first round:}
\end{center}
\begin{enumerate}[label=$(\roman*)$, leftmargin=20pt]
\item Players in $\sigma$ sample the first message $M^{(1)}_P,M^{(1)}_F$ of {principal}  and {fooling} blocks \emph{publicly}. 
\item Then, they will sample an index $i \in [p_r]$ uniformly and let $\Ins_i = \Insstar$ in the instance $\Ins$; thus, players in $\sigma$ will play the role of principal block $\pb_i$ in $\prot_r$ from now on. 
\item Next, they sample $\Jns_{i,*}$ conditioned only on $M^{(1)}_{P}$ and $\Ins_i = \Insstar$ using \emph{private} randomness by 
each vertex $v$ of $\sigma$ independently sampling $\Jns_{i,*}(v)$ only conditioned on $\Ins_i(v) , M^{(1)}_{P}$. 
\end{enumerate}
\end{tbox}
Let us argue that the joint distribution of obtained random variables at this point is close to that of the actual distribution induced by $\prot_r$ (similar to~\Cref{eq:anrw-right,eq:anrw-alg} for~\cite{AlonNRW15}):
\vspace{-10pt}
\begin{itemize}[leftmargin=15pt]
	\item The \emph{right} distribution of the  underlying variables for $\prot_r$ can be expressed as: 
	\begin{align}
		(\rM^{(1)}_P,\rM^{(1)}_F) \times (\rIns_i \mid \rM^{(1)}_P,\rM^{(1)}_F) \times (\rJns_{i,*} \mid \rIns_i , \rM^{(1)}_P,\rM^{(1)}_F). \label{eq:ours-right-1}
	\end{align}
	\vspace{-20pt}
	\item The distribution sampled from in the protocol $\sigma$ is: 
		\begin{align}
		\underbrace{(\rM^{(1)}_P,\rM^{(1)}_F)}_{\text{publicly}} \times \underbrace{\rIns_i}_{\text{input}} \times ({\underbrace{{\bigtimes}_{\!v} \rJns_{i,*}(v) \mid \rIns_i(v) , \rM^{(1)}_{P}}_{\text{privately}}}). \label{eq:ours-alg-1}
	\end{align}
\end{itemize} 
	\vspace{-10pt}
The first terms are the same. For the second terms, similar to~\Cref{eq:anrw-fooled}, we have, 
\begin{align}
	\tvd{\rIns_i}{(\rIns_i \mid \rM^{(1)}_P,\rM^{(1)}_F)}^2 \leq	 \mi{\rIns_i}{\rM^{(1)}_P,\rM^{(1)}_F} =  \mi{\rIns_i}{\rM^{(1)}_P} + \mi{\rIns_i}{\rM^{(1)}_F \mid \rM^{(1)}_P}, \label{eq:ours1}
\end{align}
using the chain rule of mutual information (\itfacts{chain-rule}) in the equality. The first term in RHS above can still be bounded by $o(1)$ by the same logic that principal blocks are oblivious to identity of
principal instance edges in their input. But such a statement is not true about fooling blocks in the second term, as those vertices themselves are not fooled. Consider the following $1$-bit protocol.

\vspace{-5pt}
\begin{example}
	Suppose we {direct} each edge of the graph randomly to one of its endpoints using public randomness. 
	Principal blocks\footnote{A player can  know whether it is principal or fooling simply based on its degree.} send the XOR of their \emph{outgoing} edges and fooling blocks send the XOR of their \emph{incoming} edges incident on $\Jns_{i,*}$ 
	for some $i \in [p_r]$. Taking the XOR of messages sent by $\pb_i$, $M^{(1)}_{P,i}$, and fooling blocks, $M^{(1)}_F$, 
	reveals XOR of all edges inside $\Ins_i$ as each such edge will be outgoing for exactly one endpoint and edges in $\Jns_{-i,*}$ cancel out in this XOR. This reveals one bit of information about $\Ins_i$, making $\mi{\rIns_i}{\rM^{(1)}_F \mid \rM^{(1)}_P} \geq 1$.  (Ideas like this are used in actual distributed sketching protocols, e.g., in~\cite{AhnGM12a,KapralovLMMS14}.) 
\end{example}
\vspace{-5pt}
\noindent
Instead, we show that fooling blocks cannot reveal much about $\Ins_i$ for an \emph{average} $i \in [p_r]$: 
\begin{align}
	\Exp_i [\mi{\rIns_i}{\rM^{(1)}_F \mid \rM^{(1)}_P}] \leq \frac{1}{p_r} \cdot \mi{\rIns_1,\ldots,\rIns_{p_r}}{\rM^{(1)}_F \mid \rM^{(1)}_P}  \leq o(1), \label{eq:be-bad}
\end{align}
where in the second inequality we used the fact that the $\polylog{(n)}$-bit messages of \emph{all} $f_r \approx n^{2/5}$ fooling blocks of size $\approx n^{1/5}$ cannot reveal more than $o(p_r)$ information as $p_r \approx n^{4/5}$ (this idea 
is similar to the ``public-vs-private'' vertices of~\cite{AssadiKO20} for one-round lower bounds in the distributed sketching model). 
This allows us to bound the LHS of~\Cref{eq:ours1} on average for $i \in [p_r]$. A similar type of argument can be applied to the third terms also to ``drop'' the conditioning on $M^{(1)}_F$, while changing the distribution only 
by $o(1)$ in total variation distance. This implies that
\[
	\Exp_i \tvd{(\rJns_{i,*} \mid \rIns_i , \rM^{(1)}_P,\rM^{(1)}_F)}{(\rJns_{i,*} \mid \rIns_i , \rM^{(1)}_P)} \leq o(1).
\]
%
By~\cite{AlonNRW15}, the second distribution here matches the product distribution sampled privately by the players (the third term of~\Cref{eq:ours-alg-1}). 
This is now sufficient for simulating the \emph{first round} of $\prot_r$ (almost) faithfully with no communication as $i \in [p_r]$ is also chosen randomly in the embedding\footnote{\!\!\cite{AlonNRW15} also works
with a random $i \in [p_r]$ but only to ensure that the underlying instance $\Ins_i$ needs to be solved by $\prot_r$ as most but not all principal instances are solved in $\prot_r$ -- \emph{all} information-theoretic guarantees 
for $\prot_r$ mentioned for the embedding of~\cite{AlonNRW15} hold for \emph{arbitrary} $i \in [p_r]$ unlike ours.}.

It is tempting  to consider our job done as we successfully simulated the first round of $\prot_r$ with no communication, and thus we \emph{eliminated} a round. But in fact, this 
is just the start of \textbf{the unique challenges of our model}. Unlike~\cite{AlonNRW15}, it is not clear how we can continue running $\prot_r$ in the subsequent rounds: 
in $\sigma$, we have only decided on the input of principal block $\pb_i$ in $\Ins$ -- the input to  other principal blocks and all fooling blocks are still undecided, and so $\prot_r$ is not well defined for the subsequent rounds. 
We now need to deviate entirely from~\cite{AlonNRW15} to handle this. 

\subsubsection{Idea Three: Partial-Input Embedding and Non-Simultaneous Simulation}\label{sec:idea3}

To continue running $\prot_r$ from its second round onwards, we should be able to simulate \emph{all} players in $\Ins$, not only the principal block $\pb_i$ responsible for $\Ins_i = \Insstar$. 
Let us consider a standard approach. 


\paragraph{Standard approach for handling remaining instances.} The standard approach is to  sample input of remaining players in $\prot_r$ using public randomness and let the ``actual'' players of $\sigma$ simulate them ``in their head'' with no communication (this corresponds to step $(iv)$ of embedding of~\cite{AlonNRW15}). 
This approach fails completely for us. Consider the fooling blocks first: at this point in the protocol $\sigma$, the players have sampled $\Jns_{i,*}$ \emph{privately} which was necessary in the first round (given the correlation of $\Jns_{i,*}(v)$ with $\Ins_i(v)$ via $M^{(1)}_{P}$ and that $\Ins_i(v)$ was only known to $v$). But given that the other endpoints of these edges are in fooling blocks, this means 
that \emph{no} single player of $\sigma$ can  even know the edges incident on a single vertex in fooling blocks, leaving no player to simulate players of $\prot_r$ in fooling blocks (or sampling rest of their inputs). 

A more subtle issue happens when it comes to the rest of principal blocks, which on the surface, should be fine given they share no edges with principal block $\pb_i$. To be able to sample instances $\Ins_{-i},\Jns_{-i,*}$ publicly in the last step of embedding,  we need the following two distributions to be close: 
\begin{align*}
	\underbrace{(\rIns_{-i}, \rJns_{-i,*} \mid \rJns_{i,*} , \rIns_i , \rM^{(1)}_P, \rM^{(1)}_F)}_{\text{right distribution}} \qquad vs. \qquad
	\underbrace{(\rIns_{-i}, \rJns_{-i,*} \mid \rM^{(1)}_P, \rM^{(1)}_F)}_{\text{``input-sampling''-protocol distribution}} \hspace{-20pt}.
\end{align*}
Yet, even a $1$-bit communication protocol can turn these two distributions far from each other: 
\begin{example}
	Suppose principal blocks remain silent and each fooling block sends the XOR of  their incident edges. Then conditioned on the messages $M^{(1)}_F$, once 
	we additionally know $\Jns_{i,*}$, we learn the parity of edges in $\Jns_{-i,*}$ which changes the distribution of $\Jns_{-i,*}$ by $\Omega(1)$.  
\end{example}

All in all, when it comes to our edge-sharing model, the standard approach of sampling the remaining instances inherently fails: $(i)$ fooling blocks are directly incident on edges in $\Jns_{i,*}$ which are part of the input to players in $\pb_i$ in $\prot_r$; $(ii)$ worse yet, the messages of fooling blocks even correlate inputs of the rest of principal vertices with those of $\pb_i$, meaning that \emph{all} principal players can reveal information about $\Ins_i$ not only the ones in $\pb_i$ that are directly incident on it.  

\paragraph{Our approach for handling remaining instances.} A key idea we use in the rest of our protocol is what we call \textbf{partial-input embedding}: we only generate the rest of the input for players $\pb_i$ and for all the remaining 
players, we will simulate them solely by \emph{sampling their messages} without ever committing to their input. Thus, our embedding keeps going even beyond the first round as we will need 
to generate the messages of remaining players throughout the entire execution of $\prot_r$. 

In particular, after running the embedding part of the first round, for any round $t > 1$, 
the players in the protocol $\sigma$ will simulate the $t$-th round of $\prot_r$ as follows: 
\begin{tbox}
\vspace{-5pt}
\begin{center}
\underline{Our embedding argument -- after first round:}
\end{center}
\begin{enumerate}[label=$(\roman*)$, leftmargin=20pt]
\item The players in $\sigma$ communicate messages of $\pb_i$ using the current content of the blackboard $M^{(<t)}$, and their inputs $\Jns_{i,*},\Ins_i$ sampled for the first round, and send the messages $M^{(t)}_{P,i}$.
\item  \emph{After} this message is revealed, the players use \emph{public} randomness to sample the $t$-th message of remaining players $M^{(t)}_{-i}:=(M^{(t)}_{P,-i},M^{(t)}_F)$ conditioned on public knowledge $M^{(<t)}, M^{(t)}_{P,i}$.  
\end{enumerate}
\end{tbox}
It is worth pointing out a rather strange aspect of this embedding. In $\prot_r$ itself, the messages $M^{(t)}_{P,i}$ and $M^{(t)}_{-i}$ are communicated \emph{simultaneously} with each other. Yet, in our simulation of $\prot_r$, 
we are  crucially using messages principal block $\pb_i$ to help us generate the remaining messages! We will discuss the necessity of this \textbf{non-simultaneous simulation} of a round in the next subsection. 

 As before, let us examine the underlying distributions in the first $t$ rounds for $t > 1$: 
\vspace{-5pt}
\begin{itemize}[leftmargin=15pt]
	\item The {right} distribution of the underlying variables up until this point in $\prot_r$ is: 
	\begin{align}
		\underbrace{(\rM^{(<t)},\rJns_{i,*}, \rIns_i)}_{\text{prior rounds}} \times (\rM^{(t)}_{P,i} \mid \rM^{(<t)}, \rJns_{i,*}, \rIns_i) \times (\rM^{(t)}_{-i} \mid \rM^{(t)}_{P,i} , \rM^{(<t)}, \rJns_{i,*}, \rIns_i). \label{eq:ours-right-2}
	\end{align}
	\vspace{-20pt}
	\item The distribution sampled from in the protocol $\sigma$ is: 
		\begin{align}
		\underbrace{(\rM^{(<t)},\rJns_{i,*}, \rIns_i)}_{\text{prior rounds}} \times \underbrace{({\bigtimes}_{\!v} \rM^{(t)}_{P,i}(v) \mid \rM^{(<t)}, \rJns_{i,*}(v), \rIns_i(v))}_{\text{communication}} \times \underbrace{(\rM^{(t)}_{-i} \mid \rM^{(t)}_{P,i} , \rM^{(<t)})}_{\text{publicly}}. \label{eq:ours-alg-2}
	\end{align}
\end{itemize}
\vspace{-7pt}

The first terms can be shown to be $o(1)$-close inductively (with base case being success of our simulation in the first round). 
The second terms are identical since the messages $M^{(t)}_{P,i}$ in $\prot_r$ are simply generated simultaneously by each vertex $v \in \pb_i$ looking at its own neighborhood 
$\Jns_{i,*}(v), \Ins_i(v)$ and the blackboard $M^{(<t)}$. For the last terms to be close, similar to~\Cref{eq:anrw-fooled,eq:ours1}, we need to bound the mutual information between $M^{(t)}_{-i}$ and $\Jns_{i,*}, \Ins_i$ 
at this point of the protocol, namely: 
\begin{align}
	 \tvd{(\rM^{(t)}_{-i} \mid \rM^{(t)}_{P,i} , \rM^{(<t)})}{(\rM^{(t)}_{-i} \mid \rM^{(t)}_{P,i}, \rM^{(<t)}, \rJns_{i,*}, \rIns_i)}^2  
	 \leq  
	 {\mi{\rM^{(t)}_{-i}}{\rJns_{i,*}, \rIns_i \mid \rM^{(t)}_{P,i}, \rM^{(<t)}}}. \label{eq:later1} 
\end{align}
Yet, while the RHS of this equation may seem similar to that of~\Cref{eq:be-bad}, this is a much more challenging term to bound as we shall discuss in the next subsection. For now, we only 
mention that our proof eventually bounds this information term \emph{on average} for $i \in [p_r]$ with $o(1)$ which allows us to continue the simulation. 

Having shown the $o(1)$-closeness of the distribution of $\prot_r$ and the one used in our embedding, the proof ends as follows. The players of $\sigma$ can continue running $\prot_r$ by playing the role of principal block $\pb_i$ in $\prot_r$ explicitly with proper communication and {keep sampling} messages of remaining players as done in the embedding. At the end of the last round, they will obtain an almost faithful simulation of the entire protocol $\prot_r$ which allows them to solve $\Insstar = \Ins_i$ as
$\prot_r$ likely needs to solve $\Ins_i$ for a random $i \in [p_r]$.  
This will then give us an $(r-1)$-round protocol for $\Insstar$ which in turn allows us to use the inductive hardness of these instances to infer the lower bound for $r$-round protocols. 

\subsubsection{Idea Four: Bounding Gradual Correlation of Players' Inputs}\label{sec:idea4}

The main technical part of our proof is to bound the information term in the RHS of~\Cref{eq:later1}, namely, the information other players can reveal about the input of principal block $\pb_i$ in a single round. By the definition of 
$M^{(t)}_{-i} = (M^{(t)}_{P,-i},M^{(t)}_F)$ and chain rule (\itfacts{chain-rule}), we have,  
\begin{align}
	\text{RHS of \Cref{eq:later1}} = \mi{\rM^{(t)}_{P,-i}}{\rJns_{i,*}, \rIns_i \mid \rM^{(t)}_{P,i} , \rM^{(<t)}} + \mi{\rM^{(t)}_F}{\rJns_{i,*}, \rIns_i \mid \rM^{(t)}_{P} , \rM^{(<t)}}. \label{eq:chain2}
\end{align}
Recall that by the construction of the instance $\Ins$, we have $\rJns_{i,*}, \rIns_i \perp \rJns_{-i,*}, \rIns_{-i}$. By the rectangle property of communication protocols, if the input of players are independent of each other, then even after communication, 
their corresponding input remains independent. \emph{Assuming} we have this conditional independence here, one can easily prove both of the following properties: 
\begin{align*}
	&\mi{\rM^{(t)}_{P,-i}}{\rJns_{i,*}, \rIns_i \mid \rM^{(t)}_{P,i} , \rM^{(<t)}}  = 0 \tag{by \itfacts{info-zero}}, \\
	&\Exp_i[\mi{\rM^{(t)}_F}{\rJns_{i,*}, \rIns_i \mid \rM^{(t)}_{P} , \rM^{(<t)}}] \leq \frac{1}{p_r} \cdot \mi{\rM^{(t)}_F}{\rJns,\rIns \mid \rM^{(t)}_{P} , \rM^{(<t)}} \leq o(1). \tag{similar to~\Cref{eq:be-bad}} 
\end{align*}
So then what is the problem here? \textbf{Short answer: edge-sharing between the players!} 

While $\rJns_{i,*}, \rIns_i \perp \rJns_{-i,*}, \rIns_{-i}$ is true initially, having fooling blocks that are able to see (subsets of) \emph{both} these sets from the other endpoints, means that
their messages can \emph{correlate} these inputs as well. In other words, it can be that $\rJns_{i,*}, \rIns_i \not\perp \rJns_{-i,*}, \rIns_{-i} \mid \rM^{(<t)}_F$ already from the second round. What is even more problematic 
is that even principal blocks in $\pb_i$ and $\pb_{-i}$ will see messages of these fooling blocks, so after the second round, even messages of other principal blocks correlate their originally independent inputs -- more formally, 
this means that $\rJns_{i,*}, \rIns_i \not\perp \rJns_{-i,*}, \rIns_{-i} \mid \rM^{(t)}_P$ (with no direct conditioning on fooling blocks' messages) can also happen after the second round! 

The following example helps to motivate our approach. 

\begin{example}
	Consider the following two protocols: 
	\vspace{-5pt}
	\begin{itemize}
		\item Protocol 1: in the second round, every principal block \emph{except for $\pb_i$} sends XOR of their edges to fooling blocks\footnote{Identity of fooling blocks can be known to everyone in the second round.} $\Jns_{-i,*}$, while fooling blocks send XOR of all their edges in $\Jns$. 
		\item Protocol 2: in the second round, every principal block sends XOR of their edges in $\Jns$ while fooling blocks send XOR of all their edges in $\Jns$. 
	\end{itemize}
	\vspace{-5pt}
	In the first protocol, conditioned on $M^{(2)}_F$, the messages $M^{(2)}_{P,-i}$ reveal the XOR of edges in $\Jns_{i,*}$, and thus the first mutual information term in~\Cref{eq:chain2} is $1$ bit (note that here $M^{(2)}_{P,i} = \emptyset$). 

\noindent
	In the second protocol, while $M^{(2)}_{P,-i}, M^{(2)}_F$ still reveal the XOR of $\Jns_{i,*}$, given that $M^{(2)}_{P,i}$ is already this XOR itself, the mutual information term in~\Cref{eq:chain2} is $0$ bit. 
\end{example}

This example shows that one can have protocols that for \emph{some} values of $i \in [p_r]$, principal blocks in $\pb_{-i}$ can reveal non-trivial information about inputs of a principal block $\pb_i$ also. But the given protocol (Protocol 1) 
is quite sensitive to the choice of index $i$, and for other indices $j \neq i$, this revealing of information no longer happens in this specific protocol. On the other hand, making the protocol less sensitive to the choice of $i$ by ``symmetrizing'' the actions of  players breaks its information-revealing property as players in $\pb_i$ themselves will reveal the information offered by others. We exploit this by bounding the first term of~\Cref{eq:chain2} \emph{on average} for $i \in [p_r]$. 
Note that this is precisely the step that our non-simultaneous simulation of a round, alluded to in~\Cref{sec:idea3}, kicks in: the messages of $M^{(2)}_{P,-i}$ are still correlated heavily with $\Jns_{i,*},\Ins_i$ even in Protocol 2; but conditioning on $M^{(2)}_{P,i}$ allows us to ``break'' this correlation and thus generate these messages even in the absence of public knowledge of $\Jns_{i,*},\Ins_i$. We argue this is true for all protocols in the following. 

To continue, by using chain rule (\itfacts{chain-rule}) on the first term of~\Cref{eq:chain2}, we get that, 
\begin{align}
	\mi{\rM^{(t)}_{P,-i}}{\rJns_{i,*}, \rIns_i \mid \rM^{(t)}_{P,i} , \rM^{(<t)}} &= \mi{\rM^{(<t)} , \rM^{(t)}_{P}}{\rJns_{i,*}, \rIns_i} - \mi{\rM^{(<t)} , \rM^{(t)}_{P,i}}{\rJns_{i,*}, \rIns_i} \label{eq:chain3} 
\end{align}
where RHS is all the information revealed by the protocol about $\Jns_{i,*},\Ins_i$ \emph{minus} the information revealed already by players $\pb_i$ and content of the blackboard.
Now, in the absence of any conditioning, one can use the fact that $\rJns_{i,*},\rIns_i \perp \rJns_{-i,*},\rIns_{i}$ to bound: 
\[
	\text{First term of~\Cref{eq:chain3} on average:} \quad \Exp_i[\mi{\rM^{(<t)},\rM^{(t)}_{P}}{\rJns_{i,*}, \rIns_i}] \leq o(1) + \frac{1}{p_r} \cdot \mi{\rM^{(\leq t)}_{P}}{\rJns, \rIns \mid \rM^{(<t)}_F},
\]
i.e., argue that fooling blocks can only reveal $o(1)$ bits about the input of an average principal block and the rest is the \emph{average} information revealed by principal blocks themselves about the entire input. The second term of~\Cref{eq:chain3} is lower bounded by (via a simple application of chain rule and non-negativity of mutual information),
\[
	\text{Second term of~\Cref{eq:chain3} on average:} \quad \Exp_i[\mi{\rM^{(<t)},\rM^{(t)}_{P,i}}{\rJns_{i,*}, \rIns_i}] \geq \Exp_i[\mi{\rM^{(\leq t)}_{P}}{\rJns_{i,*}, \rIns_i} \mid \rM^{(<t)}_F].
\]
Last step of the proof is to bound the second terms of the two equations above by showing that 
\[
	 \mi{\rM^{(\leq t)}_{P}}{\rJns, \rIns \mid \rM^{(<t)}_F} \leq \sum_{i=1}^{p_r} \mi{\rM^{(\leq t)}_{P}}{\rJns_{i,*}, \rIns_i \mid \rM^{(<t)}_F}.
\]
In words, this means that the total information revealed by principal blocks about the entire instance is bounded by the sum of the information revealed by them about each individual principal block's input $\Jns_{i,*},\Ins_{i}$ for $i \in [p_r]$ \emph{after} we condition 
on the messages of fooling blocks. This step requires a detailed calculation that at its core boils down to the fact that once we condition on $M^{(<t)}_F$, we can ``isolate'' the information revealed by each message $M^{(t)}_{P,i}$ solely to $\Jns_{i,*},\Ins_i$
-- in other words, the principal blocks cannot generate correlation with other principal blocks' inputs on their own beyond what is already forced by fooling blocks. 

Plugging in these bounds all together in~\Cref{eq:chain3} bounds the RHS by $o(1)$. A similar exercise, allows us to bound the second term in~\Cref{eq:chain2} by $o(1)$ also, which bounds the total information revealed 
about $\Jns_{i,*},\Ins_i$ by players other than the ones in $\pb_i$ by $o(1)$. This concludes the $o(1)$ bound on the mutual information term in~\Cref{eq:later1}, and implies the correctness of our simulation. 

\bigskip
To conclude, we managed to simulate \emph{all} rounds of $\prot_r$ almost faithfully by continuing the embedding throughout the protocol and as a result solve the underlying instance $\Insstar$ in $(r-1)$ rounds
using a protocol with $\polylog{(n)}$-size messages. We can now repeat this argument for $(r-1)$-round protocols and since in each recursion, the size of underlying instances drops by a factor of $\approx n^{1/5}$, we will end up with a non-trivial instance for any $r=o(\log\log{n})$ that needs to be solved by a $0$-round protocol -- a contradiction that implies our desired lower bound. 


\clearpage


\section{A Hard Distribution for Maximal Independent Set}\label{sec:mis-dist}

The following is a formal restatement of~\Cref{res:mis}.

\begin{theorem}[\Cref{res:mis}, formal]\label{thm:mis}
	For $r \ge 0$ and any $r$-round multi-party protocol (deterministic or randomized) in the shared blackboard model for computing a maximal independent set on $n$-vertex graphs with constant error probability, there must exist \emph{some} vertex communicating at least $\Omega(n^{1/\const^{r+1}})$ bits in \emph{some} round.
\end{theorem}

In this section, we give a recursive definition of the hard distribution for maximal independent set that we are going to use for our proofs in~\Cref{sec:mis-proof}.
The base case is the following hard distribution $\misHard[0]$ for protocols without any communication.

\begin{Distribution}\label{dist:mis-hard-0}
	The hard distribution $\misHard[0]$ for protocols computing a maximal independent set without any communication.

	\textbf{Parameters:} bandwidth $k$, number of vertices $n_0 = 2k$.
	\begin{enumerate}
		\item Let $E$ be an arbitrary, \emph{fixed} perfect matching over $n_0$ vertices.
		\item For $e \in E$, drop $e$ with probability $1/2$ independently.
		\item Return the graph $G$ sampled above.
	\end{enumerate}
\end{Distribution}

An immediate observation about $\misHard[0]$ is that any valid maximal independent set uniquely determines the set of matching edges that is dropped from $E$: for $e = (u,v) \in E$, $e$ is dropped from $E$ if and only if both of $u,v$ are present in the maximal independent set.
So for any deterministic referee, it can output a valid maximal independent set with probability at most $2^{-k}$ over $\misHard[0]$ if it gets no information from the vertices.
Note that this distributional bound naturally generalizes to randomized referees by an averaging argument, which is summarized in the following lemma.

\begin{lemma}[Base Case]\label{lem:mis-base-case}
	Any $0$-round protocol for computing a maximal independent set can only succeed with probability $2^{-k}$ over $\misHard[0]$.
\end{lemma}

Building upon $\misHard[0]$, we construct the $r$-round hard distribution $\misHard[r]$ recursively.
Assume we are given the $(r-1)$-round hard distribution $\misHard[r-1]$ over $n_{r-1}$ vertices.
The construction consists of two steps: first defining an auxiliary ``half distribution'' $\misHalf[r]$ and then using $\misHalf[r]$ to get the desired $\misHard[r]$, as shown below.
The ``half instances'' roughly correspond to the hard instances we talk about in~\Cref{sec:overview}.
See~\Cref{fig:mis} for an illustration.

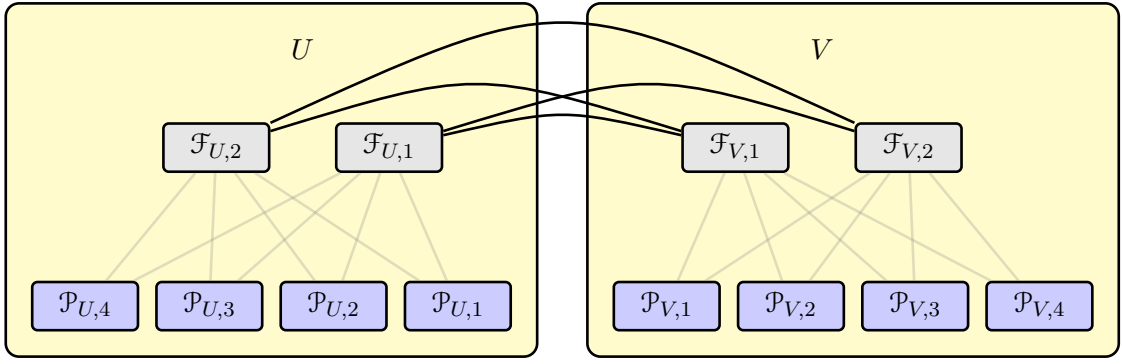
\begin{figure}[htb]
	\centering

\begin{tikzpicture}

\tikzset{layer/.style={rectangle, rounded corners=5pt, draw, black, line width=1pt,  fill=black!10, inner sep=4pt}}
\tikzset{vertex/.style={circle, ForestGreen, fill=white, line width=2pt, draw, minimum width=8pt, minimum height=8pt, inner sep=0pt}}
\tikzset{choose/.style={rectangle, line width=1pt, rounded corners = 2pt, draw, minimum width=40pt, minimum height=16pt, fill=black!10}}

\begin{scope}[local bounding box=lbb]

\node (U) {$U$};
\node[choose](FU1)[below right=20pt and 4pt of U]{$\fb_{U,1}$};
\node[choose](FU2)[below left=20pt and 4pt of U]{$\fb_{U,2}$};
\node[choose, fill=blue!20](PU1)[below right=80pt and 30pt of U]{$\pb_{U,1}$};

\foreach \i in {2,...,4}
{
	\pgfmathtruncatemacro{\ip}{\i-1};
	\node[choose, fill=blue!20] (PU\i) [left=6pt of PU\ip]{$\pb_{U,\i}$};
}

\foreach \i in {1,...,4}
{
	\foreach \j in {1,2}
	{
		\draw[line width=1pt, gray, opacity=0.25] (PU\i) -- (FU\j);
	}
}

\end{scope}

\begin{scope}[local bounding box=rbb]

\node (V) [right=180pt of U]{$V$};
\node[choose](FV1)[below left=20pt and 4pt of V]{$\fb_{V,1}$};
\node[choose](FV2)[below right=20pt and 4pt of V]{$\fb_{V,2}$};
\node[choose, fill=blue!20](PV1)[below left=80pt and 30pt of V]{$\pb_{V,1}$};

\foreach \i in {2,...,4}
{
	\pgfmathtruncatemacro{\ip}{\i-1};
	\node[choose, fill=blue!20] (PV\i) [right=6pt of PV\ip]{$\pb_{V,\i}$};
}

\foreach \i in {1,...,4}
{
	\foreach \j in {1,2}
	{
		\draw[line width=1pt, gray, opacity=0.25] (PV\i) -- (FV\j);
	}
}

\end{scope}

\begin{scope}[on background layer]
\draw[layer, fill=yellow!25] ($(lbb.south west) - (10pt,10pt)$) rectangle ($(lbb.north east) + (10pt,10pt)$);
\draw[layer, fill=yellow!25] ($(rbb.south west) - (10pt,10pt)$) rectangle ($(rbb.north east) + (10pt,10pt)$);
\end{scope}

\foreach \j in {1,2}
{
	\foreach \k in {1,2}
	{
		\pgfmathtruncatemacro{\l}{\j*\k*15};
		\draw[line width=1pt, black] (FU\j) .. controls ($(FU\j)!0.5!(FV\k) + (0,+\l pt)$) .. (FV\k);
	}
}

\end{tikzpicture}
	\caption{An illustration of our lower bound instances for maximal independent set with parameters $\halff_r = 2$ and $\halfp_r = 4$.
	The bottom vertices (blue) are principal blocks, while top vertices (gray) are fooling blocks.
	The heavy (solid black) edges fully connect fooling vertices from two ``half instances'' (yellow boxes).
	Note that these are the only edges across two ``half instances''.
	To find a maximal independent set in this graph, one needs to find maximal independent sets in all principal instances of at least one of ``half instances''.
	}\label{fig:mis}
\end{figure}

\begin{Distribution}\label{dist:mis-half-r}
	The ``half distribution'' $\misHalf[r]$ over graphs with vertex set $V$ ($r \ge 1$).

	\textbf{Parameters:} bandwidth $k$, number of \textbf{fooling blocks} $\halff_r = k^6 \cdot n_{r-1}^3$, number of \textbf{principal blocks} $\halfp_r = k^6 \cdot n_{r-1}^3 \cdot \halff_r$, number of vertices $\halfn_r = (n_{r-1}-1) \cdot \halff_r + n_{r-1} \cdot \halfp_r$, and vertex set $V$ with $|V| = \halfn_r$.
	\begin{enumerate}
		\item Partition $V$ into disjoint sets of vertices $\pb_1,\ldots,\pb_{\halfp_r},\fb_1,\ldots,\fb_{\halff_r}$ such that $\forall i \in [\halfp_r]: \; |\pb_i| = n_{r-1}$ and $\forall j \in [\halff_r]: \; |\fb_j| = n_{r-1}-1$. Define $\pb[V] := \bigcup_{i \in [\halfp_r]} \pb_i$ and $\fb[V] := \bigcup_{j \in [\halff_r]} \fb_j$.
		\item For $i \in [\halfp_r]$, sample an independent instance of $\misHard[r-1]$ on $\pb_i$.
		\item For $u \in \pb[V]$ and $j \in [\halff_r]$, sample an independent instance of $\misHard[r-1]$ on $\fb_j \cup \set{u}$ and only keep the edges adjacent to $u$ (dropping all the edges between vertices in $\fb_j$).
		\item Return the graph $G$ sampled above.
	\end{enumerate}
\end{Distribution}

\begin{Distribution}\label{dist:mis-hard-r}
	The hard distribution $\misHard[r]$ for $r$-round protocols computing a maximal independent set ($r \ge 1$).

	\textbf{Parameters:} bandwidth $k$, number of \textbf{fooling blocks} $f_r = 2\halff_r$, number of \textbf{principal blocks} $p_r = 2\halfp_r$, number of vertices $n_r = 2\halfn_r$.
	\begin{enumerate}
		\item Let $U$ and $V$ be two disjoint sets of vertices, each of size $\halfn_r$. Sample two independent instances of $\misHalf[r]$ on $U$ and $V$.
		\item For $u \in \fb[U]$ and $v \in \fb[V]$, add an edge $(u,v)$.
		\item Let $G'$ be the graph sampled above. Sample a uniformly random permutation $\perm$ over $U \cup V$ and return $G = \perm(G')$.
	\end{enumerate}
\end{Distribution}

\begin{remark}\label{rem:mis-dist}
	A few remarks are in order.
	\begin{enumerate}
		\item In the construction of the ``half distribution'' $\misHalf[r]$, we call the sets of vertices $\pb_1,\ldots,\pb_{\halfp_r}$ the \emph{principal blocks}, and the sets of vertices $\fb_1,\ldots,\fb_{\halff_r}$ the \emph{fooling blocks}. All vertices in $\pb[V]$ and $\fb[V]$ are the \emph{principal vertices} and the \emph{fooling vertices}, respectively.
		\item With a slight abuse of notation, we write $\perm(\pb_1),\ldots,\perm(\pb_{p_r})$ to denote all $p_r$ principal blocks of $\perm(U \cup V)$, and similarly $\perm(\fb_1),\ldots,\perm(\fb_{f_r})$ for all fooling blocks, in the construction of the hard distribution $\misHard[r]$.
		\item It is not hard to see that $n_r \le k^{\const^{r+1}}$ for $r \ge 0$. Indeed, $n_0 = 2k \le k^{\const}$ and by induction, the number of fooling blocks is $\halff_r \le k^6 \cdot k^{3 \cdot \const^r} \le k^{9 \cdot \const^r}$, the number of principal blocks is $\halfp_r \le k^6 \cdot k^{3 \cdot \const^r} \cdot \halff_r \le k^{18 \cdot \const^r}$, and thus $n_r \le 2 \cdot 2 \cdot k^{\const^r} \cdot \halfp_r \le k^{\const^{r+1}}$ for $r \ge 1$. Throughout the paper we assume the bandwidth parameter $k$ is at least some sufficiently large constant.
	\end{enumerate}
\end{remark}

One important property about $\misHard[r]$, which justifies our use of two ``half instances'', is that any valid maximal independent set for $G$ must also be maximal for the induced subgraph on either $\perm(\pb[U])$ or $\perm(\pb[V])$.
The implication is that solving a hard instance drawn from $\misHard[r]$ requires to solve at least one of the ``half instances'' drawn from $\misHalf[r]$.
Formally, we have the following claim.

\begin{claim}\label{clm:mis-solve-half}
	Let $\Gamma$ be any valid maximal independent set for a graph $G$ drawn from $\misHard[r]$.
	Then at least one of the following must hold:
	\begin{enumerate}
		\item $\Gamma \cap \perm(\pb[U])$ is a valid maximal independent set for the induced subgraph on $\perm(\pb[U])$.
		\item $\Gamma \cap \perm(\pb[V])$ is a valid maximal independent set for the induced subgraph on $\perm(\pb[V])$.
	\end{enumerate}
\end{claim}

\begin{proof}
	Without loss of generality we assume $\perm$ is simply the identity permutation throughout the proof.
	Suppose for now that $\Gamma$ contains one fooling vertex $f \in \fb[U]$.
	Note that our construction in~\Cref{dist:mis-hard-r} fully connects $\fb[U]$ to $\fb[V]$ so none of $\fb[V]$ is contained in $\Gamma$.
	Furthermore, those are the only edges between the two ``half instances'' on $U$ and $V$.
	Altogether, it shows $\pb[V]$ has no neighbor chosen by $\Gamma$.
	Since $\Gamma$ is a valid maximal independent set for $G$, its restriction to $\pb[V]$, i.e. $\Gamma \cap \pb[V]$, must be a valid maximal independent set for the induced subgraph on $\pb[V]$.

	The case is symmetric when $\Gamma$ contains one fooling vertex $f \in \fb[V]$.
	It is also not hard to see that both statements in the claim must hold if none of the fooling vertices is contained in $\Gamma$.
	This concludes the proof.
\end{proof}

Note that our construction in~\Cref{dist:mis-hard-r} has no edge between principal blocks, so~\Cref{clm:mis-solve-half} further implies that solving an $r$-round instance requires to solve at least half of the principal $(r-1)$-round instances.

\clearpage


\section{The Lower Bound for Maximal Independent Set}\label{sec:mis-proof}

We prove the following theorem in this section.
\Cref{thm:mis} is a straightforward corollary by an averaging argument, namely the easy direction of Yao's minimax principle~\cite{Yao77}.
Note that by the third statement of~\Cref{rem:mis-dist}, $n_r \le k^{\const^{r+1}}$ so we know $k \ge n_r^{1/\const^{r+1}}$.

\begin{theorem}\label{thm:mis-dist-lb}
	For $r = o(\log k)$, any $r$-round protocol for computing a maximal independent set that communicates at most $k$ bits per vertex in every round can only succeed with probability less than $0.1$ over $\misHard[r]$.
\end{theorem}

Our proof to~\Cref{thm:mis-dist-lb} for $r$-round protocols in general is by repeatedly applying the following round elimination lemma.

\begin{lemma}[Round Elimination]\label{lem:mis-round-elim}
	For $r = o(\log k)$ and $\delta \in [0,1]$,
	if there exists an $r$-round protocol for computing a maximal independent set that communicates at most $k$ bits per vertex in every round and succeeds with probability $\delta$ over $\misHard[r]$,
	then there also exists an $(r-1)$-round protocol for computing a maximal independent set that communicates at most $k$ bits per vertex in every round and succeeds with probability $\delta/2 - 1/n_{r-1}$ over $\misHard[r-1]$.
\end{lemma}

Before proving~\Cref{lem:mis-round-elim}, which is the main part of this section, we first show it easily implies~\Cref{thm:mis-dist-lb}.

\begin{proof}[Proof of~\Cref{thm:mis-dist-lb}]
	Suppose for the purpose of contradiction that there exists an $r$-round protocol that communicates at most $k$ bits per vertex in every round and that has success probability $0.1$ over $\misHard[r]$.
	Applying~\Cref{lem:mis-round-elim} for $r$ times, we obtain a $0$-round protocol having success probability
	\begin{align*}
		\frac{0.1}{2^r} - \sum_{t \in [r]} \frac{1}{2^{t-1} \cdot n_{t-1}}
		& \ge \frac{0.1}{2^r} - \frac{1}{n_0} \cdot \sum_{t \in [r]} \frac{1}{2^{t-1}} \tag{as $n_{t-1}$ is increasing}\\
		& \ge \frac{0.1}{2^r} - \frac{2}{n_0}\\
		& = \frac{1}{k^{o(1)}},
	\end{align*}
	over $\misHard[0]$, where the last step follows from the assumption $r = o(\log k)$.
	Recall that $n_0 = 2k$ so the second term above is $\Theta(1/k)$ and can be ignored.
	However, the existence of such a $0$-round protocol contradicts the lower bound of~\Cref{lem:mis-base-case}.
	This concludes the proof of the theorem.
\end{proof}

We prove~\Cref{lem:mis-round-elim} in the rest of this section.
To this end, fix any $r$-round protocol $\pi$ on $n_r$ vertices that communicates at most $k$ bits per vertex in every round and succeeds with probability $\delta$ over $\misHard[r]$.
By an averaging argument, we may assume without loss of generality $\pi$ is deterministic.
Before proceeding to the actual proof, let us first define the following random variables with respect to $\pi$ when its input is drawn from $\misHard[r]$.

\begin{itemize}
	\item $\rPerm$: the random permutation $\perm$ over $n_r$ vertices;
	\item $\rB_i$: the edges within the $i$-th principal block $\rPerm(\pb_i)$ for $i \in [p_r]$;
	\item $\rT_i$: the edges between the $i$-th principal block $\rPerm(\pb_i)$ and all fooling vertices $\rPerm(\fb[U \cup V])$ for $i \in [p_r]$;
	\item $\rG_i := (\rB_i, \rT_i)$: all edges incident to the $i$-th principal block $\rPerm(\pb_i)$ for $i \in [p_r]$ (there is no edge between principal blocks by our construction in~\Cref{dist:mis-hard-r});
	\item $\rG := (\rG_1,\ldots,\rG_{p_r})$: the set of all sampled edges (the edges between $\rPerm(\fb[U])$ and $\rPerm(\fb[V])$ are always present and thus not included here; there is no other edge between fooling blocks by our construction in~\Cref{dist:mis-hard-r});
	\item $\rM^{(t)}_{\sfP,i}$: the messages sent by the $i$-th principal block $\rPerm(\pb_i)$ in the $t$-th round for $i \in [p_r]$ and $t \in [r]$;
	\item $\rM^{(t)}_\sfP := (\rM^{(t)}_{\sfP,1},\ldots,\rM^{(t)}_{\sfP,p_r})$: the messages sent by all principal blocks in the $t$-th round for $t \in [r]$;
	\item $\rM^{(t)}_\sfF$: the messages sent by all fooling blocks in the $t$-th round for $t \in [r]$;
	\item $\rM^{(t)} := (\rM^{(t)}_\sfP, \rM^{(t)}_\sfF)$: all messages sent in the $t$-th round for $t \in [r]$.
\end{itemize}

Note that $\rM^{(< t)}$ is exactly the content of the blackboard at the beginning of the $t$-th round.
For any vertex $u \in \rPerm(\pb_i)$, we further define $\rB_i(u)$ as the subset of $\rB_i$ representing only edges incident to $u$.
$\rT_i(u), \rG_i(u)$ are similarly defined.
Let $\rM^{(t)}_{\sfP,i}(u)$ be the message sent by $u \in \rPerm(\pb_i)$ in the $t$-th round.
Fix any $\rPerm$, $\rM^{(t)}_{\sfP,i}$ is a function of $\rM^{(<t)}$ and $\rG_i$ while $\rM^{(t)}_{\sfP,i}(u)$ is only a function of $\rM^{(< t)}$ and $\rG_i(u)$.
After all $r$ rounds of communication, the referee has to output the solution based solely on $\rM^{(\le r)}$ since we have assumed $\pi$ to be deterministic.

\Cref{alg:mis} presents the complete simulation protocol for round elimination, formalizing our discussion in~\Cref{sec:overview}.
At a high level, we construct the following $(r-1)$-round (randomized) protocols $\tau_1,\ldots,\tau_{p_r}$ on $n_{r-1}$ vertices that are essentially simulating $\pi$ on $n_r$ vertices.
At the end of the proof, we will show there exists some index $i^* \in [p_r]$ such that $\tau_{i^*}$ simulates $\pi$ sufficiently well and is able to solve instances of $\misHard[r-1]$ with the desired probability.

\begin{Algorithm}\label{alg:mis}
	The $(r-1)$-round protocol $\tau_i$, for any \emph{fixed} $i \in [p_r]$, simulating $\pi$ for computing a maximal independent set.

	\begin{enumerate}
		\item Sample $\rPerm$ uniformly at random using public randomness. Identify the vertices of $\tau_i$ with $\rPerm(\pb_i)$ in $\pi$, and with a slight abuse of notation, any vertex $u$ of $\tau_i$ is used interchangeably with its counterpart in $\rPerm(\pb_i)$\footnote{At a high level, the vertices of $\tau_i$ are going to play the role of the $i$-th principal block in $\pi$ and jointly simulate all other vertices of $\pi$ using public randomness. That is, they proceed with $\pi$ as if they were $\rPerm(\pb_i)$.}.
			In addition, each vertex $u$ of $\tau_i$ identifies its input given in $\tau_i$ with $\rB_i(u)$ in $\pi$.
		\item Do the following \emph{without} any communication:
			\begin{enumerate}
				\item Sample $\rM^{(1)}_{\sfP,i}$, conditioned on $\rPerm$, using public randomness.
				\item For each vertex $u$ of $\tau_i$, independently sample $\rT_i(u)$, conditioned on $\rB_i(u),\rM^{(1)}_{\sfP,i},\rPerm$, using \emph{private} randomness.
				\item Sample $\rM^{(1)}_{\sfP,-i}$, conditioned on $\rM^{(1)}_{\sfP,i},\rPerm$, using public randomness.
				\item Sample $\rM^{(1)}_\sfF$, conditioned on $\rM^{(1)}_\sfP,\rPerm$, using public randomness.
			\end{enumerate}
		\item For every $t \in [2,r]$, do the following with one round of communication:
			\begin{enumerate}
				\item For each vertex $u$ of $\tau_i$, generate and broadcast $\rM^{(t)}_{\sfP,i}(u)$ as in $\pi$, based on $\rG_i(u),\rM^{(< t)},\rPerm$.
				\item Sample $\rM^{(t)}_{\sfP,-i}$, conditioned on $\rM^{(< t)},\rM^{(t)}_{\sfP,i},\rPerm$, using public randomness.
				\item Sample $\rM^{(t)}_\sfF$, conditioned on $\rM^{(< t)},\rM^{(t)}_\sfP,\rPerm$, using public randomness.
			\end{enumerate}
		\item Let $\Gamma$ be the output of the referee of $\pi$ when given $\rM^{(\le r)}$. The referee of $\tau_i$ finally outputs $\Gamma \cap \rPerm(\pb_i)$.
	\end{enumerate}
\end{Algorithm}

As discussed in~\Cref{sec:overview}, to prove~\Cref{lem:mis-round-elim}, our goal is to find an index $i^* \in [p_r]$ such that $\tau_{i^*}$ simulates $\pi$ almost perfectly.
Concretely, it is sufficient to have the distribution of the final blackboard $\rM^{(\le r)}$ sampled by $\tau_{i^*}$ be close to the true distribution generated by $\pi$.
These two distributions would be identical if $\tau_{i^*}$ \emph{were} able to do the sampling process in~\Cref{alg:mis} such that each random variable newly sampled in any step is drawn conditioned on all previously sampled random variables.
Unfortunately, this is impossible because $\rB_{i^*}$ is the input to $\tau_{i^*}$, which is not publicly known by all vertices: each vertex is only given the edges incident to it, essentially its ``local view''.
What $\tau_{i^*}$ can actually do is to sample new random variables conditioned on all random variables previously sampled using \emph{public} randomness.
The hope is that the joint distribution of all sampled random variables is not affected by much as $\tau_{i^*}$ drops conditioning on $\rB_{i^*}$ as well as all random variables sampled using private randomness, namely $\rT_{i^*}$ in~\Cref{alg:mis}.
In fact, we will show this is true \emph{on average} over all possible $i \in [p_r]$, and thus it is sufficient to pick the best index as $i^*$.

\Cref{tab:comp} makes a detailed comparison between the sampled distribution by $\tau_i$ and the true distribution in $\pi$.
Note that $\rB_i$ is given as the input to $\tau_{i}$ and by our construction in~\Cref{dist:mis-hard-r}, it has exactly the same distribution as any principal block in $\pi$.

\begin{table}[ht]
	\centering
	\begin{tabular}{|c|c|c|}
		\hline
		\diagbox{Sampled r.v.}{Conditioning r.v.}{Distribution} & Sampled distribution by $\tau_i$ & True distribution in $\pi$\\\hline
		$\rB_i$ &  & \\\hline
		$\rM^{(1)}_{\sfP,i}$ & \rPerm & $\rB_i,\rPerm$\\\hline
		$\rT_i(u)$ & $\rB_i(u),\rM^{(1)}_{\sfP,i},\rPerm$ & \\
		$\rT_i$ & & $\rB_i,\rM^{(1)}_{\sfP,i},\rPerm$\\\hline
		$\rM^{(t)}_{\sfP,-i}$ & $\rM^{(< t)},\rM^{(t)}_{\sfP,i},\rPerm$ & $\rG_i,\rM^{(< t)},\rM^{(t)}_{\sfP,i}$,\rPerm\\\hline
		$\rM^{(t)}_\sfF$ & $\rM^{(< t)},\rM^{(t)}_\sfP,\rPerm$ & $\rG_i,\rM^{(< t)},\rM^{(t)}_\sfP,\rPerm$\\\hline
	\end{tabular}
	\caption{Sampled distribution by $\tau_i$ v.s. True distribution in $\pi$}
	\label{tab:comp}
\end{table}

\begin{remark}\label{rem:comp}
	A couple of remarks about~\Cref{alg:mis} and~\Cref{tab:comp}.
	\begin{enumerate}
		\item In~\Cref{alg:mis}, $\rT_i(u)$ is sampled independently by each vertex $u$ using private randomness. This means the sampled $\rT_i$ in fact follows a product distribution, conditioned on $\rB_i,\rM^{(1)}_{\sfP,i},\rPerm$. We will prove in~\Cref{lem:mis-sample-1} that this generates precisely the true distribution of $\rT_i$.
		\item Recall that $\pi$ is assumed to be deterministic so $\rM^{(t)}_{\sfP,i}$ is a function of $\rG_i,\rM^{(< t)},\rPerm$. More specifically, $\rM^{(t)}_{\sfP,i}(u)$ is a function of $\rG_i(u),\rM^{(< t)},\rPerm$ for each vertex $u$. $\tau_i$ indeed generates them using this approach as shown in~\Cref{alg:mis}.
	\end{enumerate}
\end{remark}

Our next step is to prove every pair of the conditional distributions are close.
At the end, we will put them together to show the final blackboards $\rM^{(\le r)}$ are also close.
The comparison between the conditional distributions is split into three parts.
\Cref{lem:mis-sample-1} proves we can indeed sample the first round message $\rM^{(1)}_{\sfP,i}$ publicly and thus eliminate the first round of communication.
The first statement of~\Cref{rem:comp} is made precise by~\Cref{lem:mis-product-dist}.
A similar conditional decomposition lemma is established in~\cite{AlonNRW15}.
\Cref{lem:mis-sample-r} formalizes the intuition of directly sampling the messages of all other blocks.

With a slight abuse of notation, we may also use $< u$ to denote all vertices $v < u$ in $\rPerm(\pb_i)$ for some $i \in [p_r]$ that can be inferred from context.
Similarly, $-u$ is used as a shorthand for $\rPerm(\pb_i) \backslash \set{u}$.

\begin{lemma}\label{lem:mis-sample-1}
	Let $\epsilon_r = r/(k^4 \cdot n_{r-1}^2)$.
	For each $i \in [p_r]$,
	\[
		\mi{\rM^{(1)}_{\sfP,i}}{\rB_i \mid \rPerm} \le \epsilon_r.
	\]
\end{lemma}

\begin{proof}
	Assume without loss of generality that $i \in [\halfp_r]$.
	That is, we only consider the principal blocks on the side of $U$ in~\Cref{dist:mis-hard-r}.
	For any $u \in \rPerm(\pb_i)$, $\rT_i(u)$ is independent of $\rT_i(< u)$ given $\rB_i,\rPerm$ by our construction in~\Cref{dist:mis-hard-r}.
	This implies $\rG_i(u)$ and $\rG_i(< u)$ are independent conditioned on $\rB_i,\rPerm$.
	Using the second statement of~\Cref{rem:comp}, we know $\rM^{(1)}_{\sfP,i}(u)$ and $\rM^{(1)}_{\sfP,i}(< u)$ are independent conditioned on $\rB_i,\rPerm$ as well by the data processing inequality~(\itfacts{data-processing}).
	Then we can get
	\begin{align*}
		\phantom{=} & \mi{\rM^{(1)}_{\sfP,i}}{\rB_i \mid \rPerm}\\
		= {} & \sum_{u \in \rPerm(\pb_i)} \mi{\rM^{(1)}_{\sfP,i}(u)}{\rB_i \mid \rM^{(1)}_{\sfP,i}(< u),\rPerm} \tag{by the chain rule of mutual information~(\itfacts{chain-rule})}\\
		\le {} & \sum_{u \in \rPerm(\pb_i)} \mi{\rM^{(1)}_{\sfP,i}(u)}{\rB_i \mid \rPerm}. \tag{as $\rM^{(1)}_{\sfP,i}(u) \perp \rM^{(1)}_{\sfP,i}(< u) \mid \rB_i,\rPerm$ and by~\Cref{prop:info-decrease}}
	\end{align*}
	Since the vertices are symmetric, it suffices to show an individual term above is upper bounded by $\epsilon_r/n_{r-1}$ as $\rPerm(\pb_i)$ contains $n_{r-1}$ vertices.
	So we fix the vertex $u$ in the following.

	One crucial observation is that $u$ is simultaneously participating in $\halff_r+1$ independent instances drawn from $\misHard[r-1]$: the principal one with $\rPerm(\pb_i) \backslash \set{u}$ and $\halff_r$ fooling ones with each of the fooling blocks $\rPerm(\fb_j)$ for $j \in [\halff_r]$.
	Collectively these $\halff_r+1$ instances constitute $\rG_i(u)$.
	Fix an ordering $\Lambda$ for subsets of vertices with size $n_{r-1} - 1$.
	Let $\rS_1,\ldots,\rS_{\halff_r+1}$ denote these $\halff_r+1$ instances in the order consistent with $\Lambda$ and $\rS_{< j} = (\rS_1,\ldots,\rS_{j-1})$ for $j \in [\halff_r+1]$.
	Note that $\rG_i(u) = (\rS_1,\ldots,\rS_{\halff_r+1})$.
	Define $\rZ$ to be the set of all these $\halff_r+1$ \emph{blocks} of vertices, i.e. $\rZ := \set{\rPerm(\pb_i) \backslash \set{u}} \cup \set{\rPerm(\fb_j) \mid j \in [\halff_r]}$.
	We emphasize that $\rZ$ records the \emph{partition} of all $u$'s possible neighbors into $\halff_r+1$ blocks, but not which one corresponds to the principal instance.
	This is important because $\rS_j$ are mutually independent conditioned on $\rZ$ whereas they are not necessarily independent conditioned only on the set of all $u$'s possible neighbors.
	Let $\rW$ be the rank of the principal block among $\rZ$ according to the order defined by $\Lambda$, so $\rS_{\rW} = \rB_i$.
	Given $\rZ$, $\rW$ is uniformly distributed over $[\halff_r+1]$ because $\rPerm$ is a uniformly random permutation.
	Intuitively, $u$ cannot distinguish between all $\halff_r+1$ instances by itself, implying that $\rM^{(1)}_{\sfP,i}(u)$ should only reveal little information about the principal instance $\rB_i$.
	Formally, we have
	\begin{align*}
		\phantom{=} & \mi{\rM^{(1)}_{\sfP,i}(u)}{\rB_i \mid \rPerm}\\
		= {} & \mi{\rM^{(1)}_{\sfP,i}(u)}{\rS_{\rW} \mid \rPerm,\rZ,\rW} \tag{as $\rZ,\rW$ are completely determined by $\rPerm$ for any fixed $i,u$}\\
		\le {} & \mi{\rM^{(1)}_{\sfP,i}(u)}{\rS_{\rW} \mid \rZ,\rW} \tag{as $\rM^{(1)}_{\sfP,i}(u) \perp \rPerm \mid \rS_{\rW},\rZ,\rW$ and by~\Cref{prop:info-decrease}}\\
		= {} &  \sum_{j \in [\halff_r+1]} \prob{\rW = j} \cdot \mi{\rM^{(1)}_{\sfP,i}(u)}{\rS_j \mid \rZ,\rW=j}\\
		= {} & \frac{1}{\halff_r+1} \cdot \sum_{j \in [\halff_r+1]} \mi{\rM^{(1)}_{\sfP,i}(u)}{\rS_j \mid \rZ},
	\end{align*}
	as the joint distribution of $(\rM^{(1)}_{\sfP,i}(u),\rS_j,\rZ)$ is independent of the event $\rW=j$.
	Continuing,
	\begin{align*}
		\phantom{=} & \mi{\rM^{(1)}_{\sfP,i}(u)}{\rB_i \mid \rPerm}\\
		\le {} & \frac{1}{\halff_r+1} \cdot \sum_{j \in [\halff_r+1]} \mi{\rM^{(1)}_{\sfP,i}(u)}{\rS_j \mid \rZ}\\
		\le {} & \frac{1}{\halff_r+1} \cdot \sum_{j \in [\halff_r+1]} \mi{\rM^{(1)}_{\sfP,i}(u)}{\rS_j \mid \rS_{< j},\rZ} \tag{as $\rS_j \perp \rS_{< j} \mid \rZ$ and by~\Cref{prop:info-increase}}\\
		= {} & \frac{1}{\halff_r+1} \cdot \mi{\rM^{(1)}_{\sfP,i}(u)}{\rG_i(u) \mid \rZ} \tag{by the chain rule of mutual information~(\itfacts{chain-rule})}\\
		\le {} & \frac{1}{\halff_r} \cdot \en{\rM^{(1)}_{\sfP,i}(u) \mid \rZ} \tag{by the definition of mutual information and non-negativity of entropy~(\itfacts{uniform})}\\
		\le {} & \frac{1}{\halff_r} \cdot \en{\rM^{(1)}_{\sfP,i}(u)} \tag{as conditioning can only reduce entropy~(\itfacts{cond-reduce})}\\
		\le {} & \frac{k}{\halff_r}. \tag{by the assumption on $\pi$'s communication and~\itfacts{uniform}}
	\end{align*}
	Plugging in $\halff_r$ as defined in~\Cref{dist:mis-hard-r}, we finally get the desired upper bound $k/\halff_r = 1/(k^5 \cdot n_{r-1}^3) \le \epsilon_r/n_{r-1}$.
	This concludes the proof by our argument at the beginning.
\end{proof}

\begin{lemma}\label{lem:mis-product-dist}
	For each $i \in [p_r]$, and fixed $\rB_i,\rM^{(1)}_{\sfP,i},\rPerm$,
	\[
		\distribution{\rT_i \mid \rB_i,\rM^{(1)}_{\sfP,i},\rPerm} \sim \bigtimes_{u \in \rPerm(\pb_i)} \distribution{\rT_i(u) \mid \rB_i(u),\rM^{(1)}_{\sfP,i},\rPerm}.
	\]
\end{lemma}

\begin{proof}
	Let $\overline{\rB}_i(u)$ denote the subset of $\rB_i$ representing edges not incident to $u$\footnote{Note that $\overline{\rB}_i(u) \ne \rB_i(-u)$ since each edge $(u,v)$ appears in both $\rB_i(u)$ and $\rB_i(v)$.} for $u \in \rPerm(\pb_i)$.
	So $\rB_i = (\rB_i(u),\overline{\rB}_i(u))$.
	It suffices to show $\mi{\rT_i(u)}{\rT_i(-u),\overline{\rB}_i(u) \mid \rB_i(u),\rM^{(1)}_{\sfP,i},\rPerm} = 0$.
	Using the second statement of~\Cref{rem:comp}, we have
	\begin{align}
		\mi{\rM^{(1)}_{\sfP,i}(u)}{\rT_i(-u),\overline{\rB}_i(u) \mid \rT_i(u),\rB_i(u),\rM^{(1)}_{\sfP,i}(-u),\rPerm} = 0, \label{eqn:mis-independence-u}
	\end{align}
	because $\rM^{(1)}_{\sfP,i}(u)$ is completely determined by $\rG_i(u) = (\rB_i(u),\rT_i(u)),\rPerm$.
	Similarly, we also have
	\begin{align}
		\mi{\rM^{(1)}_{\sfP,i}(-u)}{\rT_i(u) \mid \rT_i(-u),\overline{\rB}_i(u),\rB_i(u),\rPerm} = 0, \label{eqn:mis-independence-minus-u}
	\end{align}
	because $\rM^{(1)}_{\sfP,i}(-u)$ is completely determined by $\rB_i = (\rB_i(u),\overline{\rB}_i(u)),\rT_i(-u),\rPerm$.
	Combining \Cref{eqn:mis-independence-u,eqn:mis-independence-minus-u}, we then get
	\begin{align*}
		\phantom{=} & \mi{\rT_i(u)}{\rT_i(-u),\overline{\rB}_i(u) \mid \rB_i(u),\rM^{(1)}_{\sfP,i},\rPerm}\\
		= {} & \mi{\rT_i(u)}{\rT_i(-u),\overline{\rB}_i(u) \mid \rB_i(u),\rM^{(1)}_{\sfP,i}(u),\rM^{(1)}_{\sfP,i}(-u),\rPerm}\\
		\le {} & \mi{\rT_i(u)}{\rT_i(-u),\overline{\rB}_i(u) \mid \rB_i(u),\rM^{(1)}_{\sfP,i}(-u),\rPerm} \tag{by~\Cref{eqn:mis-independence-u} and~\Cref{prop:info-decrease}}\\
		\le {} & \mi{\rT_i(u)}{\rT_i(-u),\overline{\rB}_i(u) \mid \rB_i(u),\rPerm} \tag{by~\Cref{eqn:mis-independence-minus-u} and~\Cref{prop:info-decrease}}\\
		= {} & 0,
	\end{align*}
	by our construction in~\Cref{dist:mis-hard-r}.
\end{proof}

\Cref{lem:mis-sample-1,lem:mis-product-dist} together ensure~\Cref{alg:mis} simulates the input and the first round of communication with little bias.
Building upon this, \Cref{lem:mis-sample-r} takes care of all remaining rounds.
This is accomplished using the novel idea of non-simultaneous simulation as discussed in~\Cref{sec:idea3}.

\begin{lemma}\label{lem:mis-sample-r}
	Let $\epsilon_r = 1/(k^4 \cdot n_{r-1}^2)$.
	For each $t \in [r]$,
	\begin{enumerate}
		\item $\Exp_{i \in [p_r]}\mi{\rM^{(t)}_{\sfP,-i}}{\rG_i \mid \rM^{(< t)},\rM^{(t)}_{\sfP,i},\rPerm} \le \epsilon_r$.
		\item $\Exp_{i \in [p_r]}\mi{\rM^{(t)}_\sfF}{\rG_i \mid \rM^{(< t)},\rM^{(t)}_\sfP,\rPerm} \le \epsilon_r$.
	\end{enumerate}
\end{lemma}

Before going into the actual proof of~\Cref{lem:mis-sample-r}, we first present the following technical claim.
Roughly, it shows what is revealed about $\rG$ as a whole is no more than the sum of the information revealed about individual $\rG_i$ by each principal block itself, justifying~\Cref{sec:idea4}.

\begin{claim}\label{clm:mis-sum-info}
	For each $t \in [r]$,
	\begin{align}
		\mi{\rM^{(\le t)}_\sfP}{\rG \mid \rM^{(< t)}_\sfF,\rPerm} \le \sum_{i \in [p_r]} \mi{\rM^{(< t)}_\sfP,\rM^{(t)}_{\sfP,i}}{\rG_i \mid \rM^{(< t)}_\sfF,\rPerm}.\label{eqn:mis-claim}
	\end{align}
\end{claim}

\begin{proof}
	The proof is by rewriting both sides of the above inequality using the chain rule of mutual information~(\itfacts{chain-rule}) for multiple times.
	For the left hand side of~\Cref{eqn:mis-claim}, we have
	\begin{align*}
		\phantom{=} & \mi{\rM^{(\le t)}_\sfP}{\rG \mid \rM^{(< t)}_\sfF,\rPerm}\\
		= {} & \sum_{t' \in [t]} \mi{\rM^{(t')}_\sfP}{\rG \mid \rM^{(< t')}_\sfP,\rM^{(< t)}_\sfF,\rPerm}\\
		= {} & \sum_{t' \in [t]} \sum_{i \in [p_r]} \mi{\rM^{(t')}_{\sfP,i}}{\rG \mid \rM^{(t')}_{\sfP,< i},\rM^{(< t')}_\sfP,\rM^{(< t)}_\sfF,\rPerm}\\
		\le {} & \sum_{t' \in [t]} \sum_{i \in [p_r]} \mi{\rM^{(t')}_{\sfP,i}}{\rG \mid \rM^{(< t')}_\sfP,\rM^{(< t)}_\sfF,\rPerm}, \tag{by~\Cref{prop:info-decrease}}
	\end{align*}
	where we use the observation that $\rM^{(t')}_{\sfP,i}$ is fully determined by $\rG,\rPerm$, as $\pi$ is deterministic, and thus conditionally independent of $\rM^{(t')}_{\sfP,< i}$.
	Continuing,
	\begin{align*}
		\phantom{=} & \mi{\rM^{(\le t)}_\sfP}{\rG \mid \rM^{(< t)}_\sfF,\rPerm}\\
		\le {} & \sum_{t' \in [t]} \sum_{i \in [p_r]} \mi{\rM^{(t')}_{\sfP,i}}{\rG \mid \rM^{(< t')}_\sfP,\rM^{(< t)}_\sfF,\rPerm}\\
		= {} & \sum_{t' \in [t]} \sum_{i \in [p_r]} \mi{\rM^{(t')}_{\sfP,i}}{\rG_i,\rG_{-i} \mid \rM^{(< t')}_\sfP,\rM^{(< t)}_\sfF,\rPerm}\\
		= {} & \sum_{t' \in [t]} \sum_{i \in [p_r]} \left[ \mi{\rM^{(t')}_{\sfP,i}}{\rG_i \mid \rM^{(< t')}_\sfP,\rM^{(< t)}_\sfF,\rPerm} + \mi{\rM^{(t')}_{\sfP,i}}{\rG_{-i} \mid \rG_i,\rM^{(< t')}_\sfP,\rM^{(< t)}_\sfF,\rPerm} \right]\\
		= {} & \sum_{t' \in [t]} \sum_{i \in [p_r]} \mi{\rM^{(t')}_{\sfP,i}}{\rG_i \mid \rM^{(< t')}_\sfP,\rM^{(< t)}_\sfF,\rPerm},
	\end{align*}
	as $\rM^{(t')}_{\sfP,i}$ is fully determined by $\rG_i,\rM^{(< t')},\rPerm$ using the second statement of~\Cref{rem:comp}.
	The right hand side of~\Cref{eqn:mis-claim} can be bounded as follows.
	\begin{align*}
		\phantom{=} & \sum_{i \in [p_r]} \mi{\rM^{(< t)}_\sfP,\rM^{(t)}_{\sfP,i}}{\rG_i \mid \rM^{(< t)}_\sfF,\rPerm}\\
		= {} & \sum_{i \in [p_r]} \mi{\rM^{(t)}_{\sfP,i}}{\rG_i \mid \rM^{(< t)}_\sfP,\rM^{(< t)}_\sfF,\rPerm} + \sum_{i \in [p_r]} \mi{\rM^{(< t)}_\sfP}{\rG_i \mid \rM^{(< t)}_\sfF,\rPerm}\\
		= {} & \sum_{i \in [p_r]} \mi{\rM^{(t)}_{\sfP,i}}{\rG_i \mid \rM^{(< t)}_\sfP,\rM^{(< t)}_\sfF,\rPerm} + \sum_{i \in [p_r]} \sum_{t' \in [t-1]} \mi{\rM^{(t')}_\sfP}{\rG_i \mid \rM^{(< t')}_\sfP,\rM^{(< t)}_\sfF,\rPerm}\\
		= {} & \sum_{i \in [p_r]} \mi{\rM^{(t)}_{\sfP,i}}{\rG_i \mid \rM^{(< t)}_\sfP,\rM^{(< t)}_\sfF,\rPerm} + \sum_{i \in [p_r]} \sum_{t' \in [t-1]} \mi{\rM^{(t')}_{\sfP,i},\rM^{(t')}_{\sfP,-i}}{\rG_i \mid \rM^{(< t')}_\sfP,\rM^{(< t)}_\sfF,\rPerm}\\
		\ge {} & \sum_{i \in [p_r]} \mi{\rM^{(t)}_{\sfP,i}}{\rG_i \mid \rM^{(< t)}_\sfP,\rM^{(< t)}_\sfF,\rPerm} + \sum_{i \in [p_r]} \sum_{t' \in [t-1]} \mi{\rM^{(t')}_{\sfP,i}}{\rG_i \mid \rM^{(< t')}_\sfP,\rM^{(< t)}_\sfF,\rPerm} \tag{by the non-negativity and chain rule of mutual information~(\itfacts{chain-rule})}\\
		= {} & \sum_{i \in [p_r]} \sum_{t' \in [t]} \mi{\rM^{(t')}_{\sfP,i}}{\rG_i \mid \rM^{(< t')}_\sfP,\rM^{(< t)}_\sfF,\rPerm}.
	\end{align*}
	We finally reach the desired conclusion by putting the two sides together.
\end{proof}

Now we are ready to prove~\Cref{lem:mis-sample-r}.

\begin{proof}[Proof of~\Cref{lem:mis-sample-r}]~
	\paragraph{Proof of the first statement:}
	Instead of bounding the expectation directly, for convenience, we are going to work with the following summation:
	\begin{align}
		\phantom{=} & \sum_{i \in [p_r]} \mi{\rM^{(t)}_{\sfP,-i}}{\rG_i \mid \rM^{(< t)},\rM^{(t)}_{\sfP,i},\rPerm}\nonumber\\
		= {} & \sum_{i \in [p_r]} \mi{\rM^{(< t)},\rM^{(t)}_\sfP}{\rG_i \mid \rPerm} - \sum_{i \in [p_r]} \mi{\rM^{(< t)},\rM^{(t)}_{\sfP,i}}{\rG_i \mid \rPerm},\label{eqn:mis-sum-pb}
	\end{align}
	as $\rM^{(t)}_\sfP = (\rM^{(t)}_{\sfP,i},\rM^{(t)}_{\sfP,-i})$ and by the chain rule of mutual information~(\itfacts{chain-rule}).
	The first term above can be upper bounded as
	\begin{align*}
		\phantom{=} & \sum_{i \in [p_r]} \mi{\rM^{(< t)},\rM^{(t)}_\sfP}{\rG_i \mid \rPerm}\\
		\le {} & \sum_{i \in [p_r]} \mi{\rM^{(< t)},\rM^{(t)}_\sfP}{\rG_i \mid \rG_{< i},\rPerm} \tag{as $\rG_i \perp \rG_{< i} \mid \rPerm$ and by~\Cref{prop:info-increase}}\\
		= {} & \mi{\rM^{(< t)},\rM^{(t)}_\sfP}{\rG \mid \rPerm} \tag{by the chain rule of mutual information~(\itfacts{chain-rule})}\\
		= {} & \mi{\rM^{(\le t)}_\sfP,\rM^{(< t)}_\sfF}{\rG \mid \rPerm}\\
		\le {} & \en{\rM^{(< t)}_\sfF} + \mi{\rM^{(\le t)}_\sfP}{\rG \mid \rM^{(< t)}_\sfF,\rPerm}. \tag{by the chain-rule of mutual information~(\itfacts{chain-rule})}
	\end{align*}
	Plugging into~\Cref{eqn:mis-sum-pb}, we have
	\begin{align*}
		\phantom{=} & \sum_{i \in [p_r]} \mi{\rM^{(t)}_{\sfP,-i}}{\rG_i \mid \rM^{(< t)},\rM^{(t)}_{\sfP,i},\rPerm}\\
		\le {} & \en{\rM^{(< t)}_\sfF} + \mi{\rM^{(\le t)}_\sfP}{\rG \mid \rM^{(< t)}_\sfF,\rPerm} - \sum_{i \in [p_r]} \mi{\rM^{(< t)},\rM^{(t)}_{\sfP,i}}{\rG_i \mid \rPerm}\\
		= {} & \en{\rM^{(< t)}_\sfF} + \mi{\rM^{(\le t)}_\sfP}{\rG \mid \rM^{(< t)}_\sfF,\rPerm} - \sum_{i \in [p_r]} \mi{\rM^{(< t)}_\sfP,\rM^{(< t)}_\sfF,\rM^{(t)}_{\sfP,i}}{\rG_i \mid \rPerm}\\
		\le {} & \en{\rM^{(< t)}_\sfF} + \mi{\rM^{(\le t)}_\sfP}{\rG \mid \rM^{(< t)}_\sfF,\rPerm} - \sum_{i \in [p_r]} \mi{\rM^{(< t)}_\sfP,\rM^{(t)}_{\sfP,i}}{\rG_i \mid \rM^{(< t)}_\sfF,\rPerm} \tag{by the non-negativity and chain rule of mutual information~(\itfacts{chain-rule})}\\
		\le {} & \en{\rM^{(< t)}_\sfF} \tag{by~\Cref{clm:mis-sum-info}}\\
		\le {} & k \cdot (n_{r-1}-1) \cdot f_r \cdot (t-1), \tag{by the subadditivity of entropy~(\itfacts{sub-additivity})}
	\end{align*}
	since there are $f_r$ fooling blocks of $n_{r-1}-1$ fooling vertices each, and every fooling vertex communicates at most $k$ bits in each of the first $t-1$ rounds.
	Going back to the expectation, we finally get
	\begin{align*}
		\phantom{=} & \Exp_{i \in [p_r]}\mi{\rM^{(t)}_{\sfP,-i}}{\rG_i \mid \rM^{(< t)},\rM^{(t)}_{\sfP,i},\rPerm}\\
		= {} & \frac{1}{p_r} \cdot \sum_{i \in [p_r]} \mi{\rM^{(t)}_{\sfP,-i}}{\rG_i \mid \rM^{(< t)},\rM^{(t)}_{\sfP,i},\rPerm}\\
		\le {} & \frac{k \cdot n_{r-1} \cdot f_r \cdot r}{p_r}\\
		\le {} & \epsilon_r,
	\end{align*}
	by the assumption $r = o(\log k)$.

	\paragraph{Proof of the second statement:}
	The proof is quite similar to the first one.
	Our goal is still to the bound the following summation:
	\begin{align}
		\phantom{=} & \sum_{i \in [p_r]} \mi{\rM^{(t)}_\sfF}{\rG_i \mid \rM^{(< t)},\rM^{(t)}_\sfP,\rPerm}\nonumber\\
		= {} & \sum_{i \in [p_r]} \mi{\rM^{(\le t)}}{\rG_i \mid \rPerm} - \sum_{i \in [p_r]} \mi{\rM^{(< t)},\rM^{(t)}_\sfP}{\rG_i \mid \rPerm},\label{eqn:mis-sum-fb}
	\end{align}
	as $\rM^{(t)} = (\rM^{(t)}_\sfP,\rM^{(t)}_\sfF)$ and by the chain rule of mutual information~(\itfacts{chain-rule}).
	Again we bound the first term above as follows.
	\begin{align*}
		\phantom{=} & \sum_{i \in [p_r]} \mi{\rM^{(\le t)}}{\rG_i \mid \rPerm}\\
		\le {} & \sum_{i \in [p_r]} \mi{\rM^{(\le t)}}{\rG_i \mid \rG_{< i},\rPerm} \tag{as $\rG_i \perp \rG_{< i} \mid \rPerm$ and by~\Cref{prop:info-increase}}\\
		= {} & \mi{\rM^{(\le t)}}{\rG \mid \rPerm} \tag{by the chain rule of mutual information~(\itfacts{chain-rule})}\\
		= {} & \mi{\rM^{(\le t)}_\sfP,\rM^{(< t)}_\sfF,\rM^{(t)}_\sfF}{\rG \mid \rPerm}\\
		\le {} & \en{\rM^{(t)}_\sfF} + \mi{\rM^{(\le t)}_\sfP,\rM^{(< t)}_\sfF}{\rG \mid \rPerm} \tag{by the chain-rule of mutual information~(\itfacts{chain-rule})}\\
		\le {} & \en{\rM^{(t)}_\sfF} + \en{\rM^{(< t)}_\sfF} + \mi{\rM^{(\le t)}_\sfP}{\rG \mid \rM^{(< t)}_\sfF,\rPerm}. \tag{by the chain-rule of mutual information~(\itfacts{chain-rule})}
	\end{align*}
	Plugging into~\Cref{eqn:mis-sum-fb}, we have
	\begin{align*}
		\phantom{=} & \sum_{i \in [p_r]} \mi{\rM^{(t)}_\sfF}{\rG_i \mid \rM^{(< t)},\rM^{(t)}_\sfP,\rPerm}\\
		\le {} & \en{\rM^{(t)}_\sfF} + \en{\rM^{(< t)}_\sfF} + \mi{\rM^{(\le t)}_\sfP}{\rG \mid \rM^{(< t)}_\sfF,\rPerm} - \sum_{i \in [p_r]} \mi{\rM^{(< t)},\rM^{(t)}_\sfP}{\rG_i \mid \rPerm}\\
		= {} & \en{\rM^{(t)}_\sfF} + \en{\rM^{(< t)}_\sfF} + \mi{\rM^{(\le t)}_\sfP}{\rG \mid \rM^{(< t)}_\sfF,\rPerm} - \sum_{i \in [p_r]} \mi{\rM^{(< t)}_\sfP,\rM^{(< t)}_\sfF,\rM^{(t)}_{\sfP,i},\rM^{(t)}_{\sfP,-i}}{\rG_i \mid \rPerm}\\
		\le {} & \en{\rM^{(t)}_\sfF} + \en{\rM^{(< t)}_\sfF} + \mi{\rM^{(\le t)}_\sfP}{\rG \mid \rM^{(< t)}_\sfF,\rPerm} - \sum_{i \in [p_r]} \mi{\rM^{(< t)}_\sfP,\rM^{(t)}_{\sfP,i}}{\rG_i \mid \rM^{(< t)}_\sfF,\rPerm}. \tag{by the non-negativity and chain rule of mutual information~(\itfacts{chain-rule})}\\
		\le {} & \en{\rM^{(t)}_\sfF} + \en{\rM^{(< t)}_\sfF}\tag{by~\Cref{clm:mis-sum-info}}\\
		\le {} & k \cdot (n_{r-1}-1) \cdot f_r \cdot t, \tag{by the subadditivity of entropy~(\itfacts{sub-additivity})}
	\end{align*}
	by counting the total communication of all fooling vertices.
	The desired upper bound on the expectation is derived similarly to the first statement.
\end{proof}

Technically, it is actually possible to prove the following similar to~\Cref{clm:mis-sum-info}:
\[
	\mi{\rM^{(t)}_\sfP}{\rG \mid \rM^{(< t)},\rPerm} \le \sum_{i \in [p_r]} \mi{\rM^{(t)}_{\sfP,i}}{\rG_i \mid \rM^{(< t)},\rPerm},
\]
which would justify the intuition we provided for~\Cref{clm:mis-sum-info} even better.
However, the proof of~\Cref{lem:mis-sample-r} is complicated by the fact that neither $\rG_i \perp \rG_{< i} \mid \rM^{(\le t)}_\sfP,\rPerm$ nor $\rG_i \perp \rG_{< i} \mid \rM^{(\le t)},\rPerm$ is true.
In general $\rM^{(t)}_\sfP$ depends on $\rM^{(< t)}_\sfF$, which in turn is able to correlate $\rG_i$ and $\rG_{-i}$.
At the core of the proof of~\Cref{lem:mis-sample-r} is applying the chain rule of mutual information~(\itfacts{chain-rule}) over all $\rG_i$.
To have the chain rule go through in the correct direction, what we need is the conditional independence between all $\rG_i$.
As a result, we are forced to rewrite the summation as in~\Cref{eqn:mis-sum-pb,eqn:mis-sum-fb} such that the conditional independence between all $\rG_i$ hold, and then conduct a more careful analysis to bound the amount of correlation caused by the messages of fooling vertices.
The current form of~\Cref{clm:mis-sum-info} turns out to be more appropriate for this purpose.

Combining~\Cref{lem:mis-sample-1,lem:mis-product-dist,lem:mis-sample-r}, the following corollary follows directly from Pinsker's inequality~(\Cref{fact:pinskers}).
It essentially captures our initial intuition that the final blackboard $\rM^{(\le r)}$ sampled by $\tau_i$ is close to the true distribution on average over all possible $i \in [p_r]$.

\begin{corollary}\label{cor:mis-tvd}
	Let $\mu$ be the true distribution for $(\rG,\rM^{(\le r)},\rPerm)$ in $\pi$ and for $i \in [p_r]$, $\mu_i$ be the marginal distribution of $\mu$ for $(\rG_i,\rM^{(\le r)},\rPerm)$.
	For $i \in [p_r]$, let $\nu_i$ be the distribution of $(\rG_i,\rM^{(\le r)},\rPerm)$ defined by
	\begin{align*}
		\nu_i(\rG_i,\rM^{(\le r)},\rPerm)
		& := \mu(\rB_i,\rPerm) \cdot \mu(\rM^{(1)}_{\sfP,i} \mid \rPerm) \cdot \prod_{u \in \rPerm(\pb_i)} \mu(\rT_i(u) \mid \rB_i(u),\rM^{(1)}_{\sfP,i},\rPerm)\\
		& \qquad \cdot \prod_{t \in [r]} \mu(\rM^{(t)}_{\sfP,-i} \mid \rM^{(< t)},\rM^{(t)}_{\sfP,i},\rPerm) \cdot \prod_{t \in [r]} \mu(\rM^{(t)}_\sfF \mid \rM^{(< t)},\rM^{(t)}_\sfP,\rPerm)\\
		& \qquad \cdot \prod_{t \in [2,r]}\prod_{u \in \rPerm(\pb_i)} \mu(\rM^{(t)}_{\sfP,i}(u) \mid \rG_i(u),\rM^{(< t)},\rPerm).
	\end{align*}
	It holds that
	\[
		\Exp_{i \in [p_r]}\Exp_{\rB_i \sim \mu}\tvd{\mu_i(\rM^{(\le r)},\rPerm \mid \rB_i)}{\nu_i(\rM^{(\le r)},\rPerm \mid \rB_i)} \le \frac{1}{k \cdot n_{r-1}}.
	\]
\end{corollary}

\begin{proof}
	Firstly, we convert the statements of~\Cref{lem:mis-sample-1,lem:mis-sample-r} to the language of total variation distance.
	For each $i \in [p_r]$, we have
	\begin{align*}
		\phantom{=} & \Exp_{(\rB_i,\rPerm) \sim \mu}\tvd{\mu_i(\rM^{(1)}_{\sfP,i} \mid \rB_i,\rPerm)}{\nu_i(\rM^{(1)}_{\sfP,i} \mid \rB_i,\rPerm)}\\
		= {} & \Exp_{(\rB_i,\rPerm) \sim \mu}\tvd{\mu(\rM^{(1)}_{\sfP,i} \mid \rB_i,\rPerm)}{\mu(\rM^{(1)}_{\sfP,i} \mid \rPerm)}\\
		\le {} & \Exp_{(\rB_i,\rPerm) \sim \mu}\sqrt{\kl{\mu(\rM^{(1)}_{\sfP,i} \mid \rB_i,\rPerm)}{\mu(\rM^{(1)}_{\sfP,i} \mid \rPerm)}} \tag{by Pinsker's inequality~(\Cref{fact:pinskers})}\\
		\le {} & \sqrt{\Exp_{(\rB_i,\rPerm) \sim \mu}\kl{\mu(\rM^{(1)}_{\sfP,i} \mid \rB_i,\rPerm)}{\mu(\rM^{(1)}_{\sfP,i} \mid \rPerm)}} \tag{by the concavity of $\sqrt{\cdot}$}\\
		= {} & \sqrt{\mi{\rM^{(1)}_{\sfP,i}}{\rB_i \mid \rPerm}} \tag{by~\Cref{fact:kl-info}}\\
		\le {} & \epsilon_r^{1/2}. \tag{by~\Cref{lem:mis-sample-1}}
	\end{align*}
	Meanwhile, for each $t \in [r]$, we can get
	\begin{align*}
		\phantom{=} & \Exp_{i \in [p_r]}\Exp_{(\rG_i,\rM^{(< t)},\rM^{(t)}_{\sfP,i},\rPerm) \sim \mu}\tvd{\mu_i(\rM^{(t)}_{\sfP,-i} \mid \rG_i,\rM^{(< t)},\rM^{(t)}_{\sfP,i},\rPerm)}{\nu_i(\rM^{(t)}_{\sfP,-i} \mid \rG_i,\rM^{(< t)},\rM^{(t)}_{\sfP,i},\rPerm)}\\
		= {} & \Exp_{i \in [p_r]}\Exp_{(\rG_i,\rM^{(< t)},\rM^{(t)}_{\sfP,i},\rPerm) \sim \mu}\tvd{\mu(\rM^{(t)}_{\sfP,-i} \mid \rG_i,\rM^{(< t)},\rM^{(t)}_{\sfP,i},\rPerm)}{\mu(\rM^{(t)}_{\sfP,-i} \mid \rM^{(< t)},\rM^{(t)}_{\sfP,i},\rPerm)}\\
		\le {} & \Exp_{i \in [p_r]}\Exp_{(\rG_i,\rM^{(< t)},\rM^{(t)}_{\sfP,i},\rPerm) \sim \mu}\sqrt{\kl{\mu(\rM^{(t)}_{\sfP,-i} \mid \rG_i,\rM^{(< t)},\rM^{(t)}_{\sfP,i},\rPerm)}{\mu(\rM^{(t)}_{\sfP,-i} \mid \rM^{(< t)},\rM^{(t)}_{\sfP,i},\rPerm)}} \tag{by Pinsker's inequality~(\Cref{fact:pinskers})}\\
		\le {} & \sqrt{\Exp_{i \in [p_r]}\Exp_{(\rG_i,\rM^{(< t)},\rM^{(t)}_{\sfP,i},\rPerm) \sim \mu}\kl{\mu(\rM^{(t)}_{\sfP,-i} \mid \rG_i,\rM^{(< t)},\rM^{(t)}_{\sfP,i},\rPerm)}{\mu(\rM^{(t)}_{\sfP,-i} \mid \rM^{(< t)},\rM^{(t)}_{\sfP,i},\rPerm)}} \tag{by the concavity of $\sqrt{\cdot}$}\\
		= {} & \sqrt{\Exp_{i \in [p_r]} \mi{\rM^{(t)}_{\sfP,-i}}{\rG_i \mid \rM^{(< t)},\rM^{(t)}_{\sfP,i},\rPerm}} \tag{by~\Cref{fact:kl-info}}\\
		\le {} & \epsilon_r^{1/2}, \tag{by~\Cref{lem:mis-sample-r}}
	\end{align*}
	and similarly
	\[
		\Exp_{i \in [p_r]}\Exp_{(\rG_i,\rM^{(< t)},\rM^{(t)}_\sfP,\rPerm) \sim \mu}\tvd{\mu_i(\rM^{(t)}_\sfF \mid \rG_i,\rM^{(< t)},\rM^{(t)}_\sfP,\rPerm)}{\nu_i(\rM^{(t)}_\sfF \mid \rG_i,\rM^{(< t)},\rM^{(t)}_\sfP,\rPerm)} \le \epsilon_r^{1/2}.
	\]
	We also trivially have
	\[
		\Exp_{(\rB_i,\rM^{(1)}_{\sfP,i},\rPerm) \sim \mu}\tvd{\mu_i(\rT_i \mid \rB_i,\rM^{(1)}_{\sfP,i},\rPerm)}{\nu_i(\rT_i \mid \rB_i,\rM^{(1)}_{\sfP,i},\rPerm)} = 0,
	\]
	by~\Cref{lem:mis-product-dist} for each $i \in [p_r]$, and
	\[
		\Exp_{(\rG_i,\rM^{(< t)},\rPerm) \sim \mu}\tvd{\mu_i(\rM^{(t)}_{\sfP,i} \mid \rG_i,\rM^{(< t)},\rPerm)}{\nu_i(\rM^{(t)}_{\sfP,i} \mid \rG_i,\rM^{(< t)},\rPerm)} = 0,
	\]
	by the second statement of~\Cref{rem:comp} for each $t \in [2,r]$.
	Additionally observe that
	\[
		\Exp_{\rB_i \sim \mu}\tvd{\mu_i(\rPerm \mid \rB_i)}{\nu_i(\rPerm \mid \rB_i)} = 0,
	\]
	since $\rPerm$ is a uniformly random permutation drawn independent of $\rB_i$.
	Combining all these conditional distributions using the chain rule of total variation distance~(\Cref{fact:tvd-chain-rule}), it holds that
	\begin{align*}
		\phantom{=} & \Exp_{i \in [p_r]}\Exp_{\rB_i \sim \mu}\tvd{\mu_i(\rM^{(\le r)},\rT_i,\rPerm \mid \rB_i)}{\nu_i(\rM^{(\le r)},\rT_i,\rPerm \mid \rB_i)}\\
		\le {} & \epsilon_r^{1/2} \cdot (2r+1)\\
		= {} & \frac{(2r+1)}{k^2 \cdot n_{r-1}}\\
		\le {} & \frac{1}{k \cdot n_{r-1}},
	\end{align*}
	by the linearity of expectation and the assumption $r = o(\log k)$.
	This concludes the proof as marginalization can never increase total variation distance~(\Cref{fact:tvd-marginal}).
\end{proof}

In~\Cref{cor:mis-tvd}, note that $\mu$ and $\mu_i$ are the true distributions in $\pi$ while $\nu_i$ is the distribution sampled by $\tau_i$.
We conclude this section by finishing the proof of~\Cref{lem:mis-round-elim}.

\begin{proof}[Proof of~\Cref{lem:mis-round-elim}]
	For each $i \in [p_r]$, define $\rO^\pi_i \in \set{0,1}$ to be $1$ if and only if the referee of $\pi$ outputs a valid maximal independent set $\Gamma$ for an $r$-round instance such that $\Gamma \cap \rPerm(\pb_i)$ is also a valid maximal independent set for the induced subgraph on $\rPerm(\pb_i)$.
	Also define $\rO^\tau_i \in \set{0,1}$ to be $1$ if and only if the referee of $\tau_i$ outputs a valid maximal independent set for an $(r-1)$-round instance.
	Recall that the referee of $\pi$ is a deterministic function of $\rM^{(\le r)}$, so for each $i \in [p_r]$, the referee of $\tau_i$ is a deterministic function of $\rM^{(\le r)}$ and $\rPerm$, by~\Cref{alg:mis}.

	Firstly imagine the idealized situation where $\tau_i$ \emph{were} able to sample $\rM^{(\le r)},\rPerm \mid \rB_i$ precisely following $\mu_i$.
	Since the marginal distribution for each principal instance in $\misHard[r]$ is the same as $\misHard[r-1]$, by the linearity of expectation we get
	\begin{align}
		\phantom{=} & \Exp_{i \in [p_r]}\Exp_{\rB_i \sim \mu_i}\Pr_{(\rM^{(\le r)},\rPerm \mid \rB_i) \sim \mu_i}\paren{\rO^\tau_i = 1}\nonumber\\
		= {} & \Exp_{i \in [p_r]}\Exp_{\rB_i \sim \mu}\Pr_{(\rM^{(\le r)},\rPerm \mid \rB_i) \sim \mu}\paren{\rO^\tau_i = 1} \tag{as $\mu_i(\rM^{(\le r)},\rB_i,\rPerm) = \mu(\rM^{(\le r)},\rB_i,\rPerm)$}\nonumber\\
		= {} & \Exp_{i \in [p_r]}\Pr_{(\rM^{(\le r)},\rG,\rPerm) \sim \mu}\paren{\rO^\pi_i = 1} \tag{by~\Cref{alg:mis}}\nonumber\\
		= {} & \Exp_{(\rG,\rPerm) \sim \mu}\Pr_{i \in [p_r]}\paren{\rO^\pi_i = 1} \tag{as $\rM^{(\le r)}$ is fully determined by $\rG,\rPerm$}\nonumber\\
		\ge {} & \delta/2,\label{eqn:mis-half}
	\end{align}
	because $\pi$ succeeds with probability $\delta$ by assumption, and conditioned on this event, at least half of the principal instances are solved by~\Cref{clm:mis-solve-half}.
	Now consider the real success probability of $\tau_i$ over $\nu_i$.
	By~\Cref{fact:tvd-small}, we have for each $i \in [p_r]$,
	\begin{align}
		\phantom{=} & \Exp_{\rB_i \sim \nu_i}\Pr_{(\rM^{(\le r)},\rPerm \mid \rB_i) \sim \nu_i}\paren{\rO^\tau_i = 1}\nonumber\\
		= {} & \Exp_{\rB_i \sim \mu_i}\Pr_{(\rM^{(\le r)},\rPerm \mid \rB_i) \sim \nu_i}\paren{\rO^\tau_i = 1} \tag{as $\nu_i(\rB_i) = \mu_i(\rB_i)$}\nonumber\\
		\ge {} & \Exp_{\rB_i \sim \mu_i}\Pr_{(\rM^{(\le r)},\rPerm \mid \rB_i) \sim \mu_i}\paren{\rO^\tau_i = 1} - \Exp_{\rB_i \sim \mu_i}\tvd{\mu_i(\rM^{(\le r)},\rPerm \mid \rB_i)}{\nu_i(\rM^{(\le r)},\rPerm \mid \rB_i)}.\label{eqn:mis-diff}
	\end{align}
	Combining~\Cref{cor:mis-tvd} with~\Cref{eqn:mis-half,eqn:mis-diff}, we finally get
	\begin{align*}
		\phantom{=} & \Exp_{i \in [p_r]}\Exp_{\rB_i \sim \nu_i}\Pr_{(\rM^{(\le r)},\rPerm \mid \rB_i) \sim \nu_i}\paren{\rO^\tau_i = 1}\\
		\ge {} & \Exp_{i \in [p_r]}\Exp_{\rB_i \sim \mu_i}\Pr_{(\rM^{(\le r)},\rPerm \mid \rB_i) \sim \mu_i}\paren{\rO^\tau_i = 1} - \Exp_{i \in [p_r]}\Exp_{\rB_i \sim \mu_i}\tvd{\mu_i(\rM^{(\le r)},\rPerm \mid \rB_i)}{\nu_i(\rM^{(\le r)},\rPerm \mid \rB_i)}\\
		\ge {} & \delta/2 - 1/n_{r-1},
	\end{align*}
	as desired.
	Picking the index $i^* \in [p_r]$ that maximizes the success probability of $\tau_{i^*}$ concludes the proof.
\end{proof}

\clearpage


\section{The Lower Bound for Approximate Bipartite Matching}\label{sec:apx}

In this section we adapt the techniques for maximal independent set to prove the following formal version of~\Cref{res:mm}.

\begin{theorem}[\Cref{res:mm}, formal]\label{thm:mm}
	For $r \ge 0$ and any $r$-round multi-party protocol (deterministic or randomized) in the shared blackboard model for computing a maximal matching or any constant factor approximation to maximum matching on $n$-vertex (bipartite) graphs, there must exist \emph{some} vertex communicating at least $\Omega(n^{1/\const^{r+1}})$ bits in \emph{some} round.
\end{theorem}

Intuitively, \Cref{dist:mis-hard-r} and~\Cref{alg:mis} make little use of any property specific to independent sets so most of the lemmas hold for matchings as well.
For convenience, we first make minor adjustment to the hard distributions in~\Cref{sec:apx-dist} to better fit the need of approximation, and then present the lower bound proof for approximate matching for \emph{general} graphs in~\Cref{sec:apx-proof}.
It is worth noticing that our constructed instances may not be bipartite in general.
Fortunately, \Cref{sec:bipartite} gives a simple reduction to the bipartite case, concluding the proof of~\Cref{thm:mm}.

\subsection{A Hard Distribution for Approximate Matching}\label{sec:apx-dist}

We use the following base case for approximate matching.
The idea is to have maximum matchings of a fixed size that remain hard to approximate.
This will help simplify the calculation in later proofs a lot.

\begin{Distribution}\label{dist:apx-hard-0}
	The hard distribution $\apxHard[0]$ for protocols computing an approximate matching without any communication.

	\textbf{Parameters:} bandwidth $k$, number of vertices $n_0 = 2k$.
	\begin{enumerate}
		\item Let $U$ and $V$ be two disjoint sets of vertices, each of size $k$. Sample two vertices $u \in U, v \in V$ uniformly at random and independently.
		\item Add an edge $(u,v)$.
		\item Return the graph $G$ sampled above.
	\end{enumerate}
\end{Distribution}

It is easy to see any graph $G$ drawn from $\apxHard[0]$ always has a maximum matching of size $1$.
Recall that protocols for approximate matching are required to output a valid matching (though potentially containing non-existing edges), which is of size at most $k$.
Since the chosen edge $(u,v)$ is sampled uniformly at random from $k^2$ possibilities, no protocols can achieve an approximation ratio better than $k^2/k = k$ if no information is revealed by the vertices.
This is summarized in the following lemma.

\begin{lemma}[Base Case]\label{lem:apx-base-case}
	Any $0$-round protocol for computing an approximate matching has an approximation ratio no better than $k$ over $\apxHard[0]$.
\end{lemma}

The construction for $r$-round hard distributions $\apxHard[r]$ is almost the same as for $\misHard[r]$.
In fact, it can be even simplified in the sense that the ``half instances'' are sufficient for the purpose of constructing a hard distribution.
Concretely, we construct $\apxHard[r]$ recursively as follows.

\begin{Distribution}\label{dist:apx-hard-r}
	The hard distribution $\apxHard[r]$ for $r$-round protocols computing an approximate matching ($r \ge 1$).

	\textbf{Parameters:} bandwidth $k$, number of fooling blocks $f_r = k^6 \cdot n_{r-1}^3$, number of principal blocks $p_r = k^6 \cdot n_{r-1}^3 \cdot f_r$, number of vertices $n_r = (n_{r-1}-1) \cdot f_r + n_{r-1} \cdot p_r$, and vertex set $V$ with $|V| = n_r$.
	\begin{enumerate}
		\item Partition $V$ into disjoint sets of vertices $\pb_1,\ldots,\pb_{p_r},\fb_1,\ldots,\fb_{f_r}$ such that $\forall i \in [p_r]: \; |\pb_i| = n_{r-1}$ and $\forall j \in [f_r]: \; |\fb_j| = n_{r-1}-1$. Define $\pb[V] := \bigcup_{i \in [p_r]} \pb_i$ and $\fb[V] := \bigcup_{j \in [f_r]} \fb_j$.
		\item For $i \in [p_r]$, sample an independent instance of $\apxHard[r-1]$ on $\pb_i$.
		\item For $u \in \pb[V]$ and $j \in [f_r]$, sample an independent instance of $\apxHard[r-1]$ on $\fb_j \cup \set{u}$ and only keep the edges adjacent to $u$ (dropping all the edges between vertices in $\fb_j$).
		\item Let $G'$ be the graph sampled above. Sample a uniformly random permutation $\perm$ over $V$ and return $G = \perm(G')$.
	\end{enumerate}
\end{Distribution}

It is not hard to verify that $n_r \le k^{\const^{r+1}}$ still holds for $r \ge 0$.
At a high level, the number of fooling vertices is rather small as $f_r \ll p_r$, so their contribution to the size of maximum matchings is limited.
On the other hand, a vast majority of matching edges should come from within the principal blocks so a good approximation ratio for $r$-round instances implies good approximation ratios over all principal $(r-1)$-round instances on average.
\Cref{clm:apx-size} provides a useful lower bound on the size of maximum matchings for graphs drawn from $\apxHard[r]$.
It will be used in the proof at the very end of this section.

\begin{claim}\label{clm:apx-size}
	Let $\Gamma$ be any valid maximum matching for a graph $G$ drawn from $\apxHard[r]$.
	Then,
	\begin{align}
		|\Gamma| \ge \frac{n_r}{2k} \cdot \paren{1 - \sum_{t \in [r]} \frac{f_t}{p_t}} \ge \frac{n_r}{4k}.\label{eqn:apx-size}
	\end{align}
\end{claim}

\begin{proof}
	The base case of $r=0$ holds trivially.
	For $r \ge 1$, we know by induction that any principal $(r-1)$-round instance has a maximum matching of size at least $\frac{n_{r-1}}{2k} \cdot (1 - \sum_{t \in [r-1]} f_t/p_t)$.
	Since all $p_r$ principal blocks are disjoint by our construction in~\Cref{dist:apx-hard-r}, we get
	\begin{align*}
		|\Gamma|
		& \ge p_r \cdot \frac{n_{r-1}}{2k} \cdot \paren{1 - \sum_{t \in [r-1]} \frac{f_t}{p_t}}\\
		& \ge \frac{p_r}{f_r+p_r} \cdot \frac{n_r}{2k} \cdot \paren{1 - \sum_{t \in [r-1]} \frac{f_t}{p_t}} \tag{as $n_r \le n_{r-1} \cdot (f_r+p_r)$}\\
		& \ge \paren{1 - \frac{f_r}{p_r}} \cdot \frac{n_r}{2k} \cdot \paren{1 - \sum_{t \in [r-1]} \frac{f_t}{p_t}}\\
		& \ge \frac{n_r}{2k} \cdot \paren{1 - \sum_{t \in [r]} \frac{f_t}{p_t}}.
	\end{align*}
	The last inequality of~\Cref{eqn:apx-size} follows from the simple fact that
	\[
		\sum_{t \in [r]} \frac{f_t}{p_t} = \sum_{t \in [r]} \frac{1}{k^6 \cdot n_{t-1}^3} \le \frac{1}{2},
	\]
	since $r = o(\log k)$ by assumption.
	This concludes the proof.
\end{proof}

\subsection{Proof of the Lower Bound for Approximate Matching}\label{sec:apx-proof}

The version of \Cref{thm:mm} for general graphs is a straightforward corollary of the following distributional lower bound by Yao's minimax principal~\cite{Yao77}.
We point out that since any maximal matching is also a $2$-approximate matching, it is sufficient to prove the hardness of approximate matching.

\begin{theorem}\label{thm:apx-dist-lb}
	For $r = o(\log k)$, any $r$-round protocol for computing an approximate matching for general graphs that communicates at most $k$ bits per vertex in every round has an approximation ratio no better than $\Omega(k)$ over $\apxHard[r]$.
\end{theorem}

The proof of~\Cref{thm:apx-dist-lb} is again via a round elimination lemma as shown in~\Cref{lem:apx-round-elim}.

\begin{lemma}[Round Elimination]\label{lem:apx-round-elim}
	For $r = o(\log k)$ and $\alpha = \omega(1/n_{r-1})$, if there exists an $r$-round protocol for computing an approximate matching that communicates at most $k$ bits per vertex in every round and has an approximation ratio of $\alpha^{-1}$ over $\apxHard[r]$, then there also exists an $(r-1)$-round protocol for computing an approximate matching that communicates at most $k$ bits per vertex in every round and has an approximation ratio of $(\alpha - C / n_{r-1})^{-1}$ over $\apxHard[r-1]$, for some universal constant $C > 0$.
\end{lemma}

Before proving~\Cref{lem:apx-round-elim}, which is the main part of this section, we first show it easily implies~\Cref{thm:apx-dist-lb}.

\begin{proof}[Proof of~\Cref{thm:apx-dist-lb}]
	Suppose for the purpose of contradiction that there exists an $r$-round protocol that communicates at most $k$ bits per vertex and that has an approximation ratio of $\alpha^{-1} = o(k)$ over $\apxHard[r]$.
	Applying~\Cref{lem:apx-round-elim} for $r$ times, we obtain a $0$-round protocol having an approximation ratio of
	\[
		\paren{\alpha - C \cdot \sum_{t \in [r]} \frac{1}{n_{t-1}}}^{-1} \le \paren{\alpha - \frac{2C}{n_0}}^{-1} = o(k),
	\]
	over $\apxHard[0]$, as $n_{t-1}$ is doubly exponentially increasing, and $\alpha = \omega(1/k), n_0 = 2k$.
	However, the existence of such a $0$-round protocol contradicts the lower bound of~\Cref{lem:apx-base-case}, concluding the proof.
\end{proof}

To prove~\Cref{lem:apx-round-elim}, we use the same approach for simulation as in~\Cref{alg:mis}.
Fix a deterministic $r$-round protocol $\pi$ on $n_r$ vertices that communicates at most $k$ bits per vertex in every round and has an approximation ratio of $\alpha^{-1}$ over $\apxHard[r]$.
We define exactly the same set of random variables as in~\Cref{sec:mis-proof} and construct the $(r-1)$-round (randomized) protocols $\tau_1,\ldots,\tau_{p_r}$ on $n_{r-1}$ vertices, which are identical to~\Cref{alg:mis} except for the processing of the final output.
Specifically, let $\Gamma$ be the output of the referee of $\pi$ when given $\rM^{(\le r)}$.
The referee of $\tau_i$ finally outputs $\Gamma \cap \paren{\rPerm(\pb_i) \times \rPerm(\pb_i)}$, namely the \emph{edges} within $\rPerm(\pb_i)$, for computing approximate matchings.
See~\Cref{alg:apx} for a recap of the complete simulation protocol for round elimination.

\begin{Algorithm}\label{alg:apx}
	The $(r-1)$-round protocol $\tau_i$, for any \emph{fixed} $i \in [p_r]$, simulating $\pi$ for computing an approximate matching.

	\begin{enumerate}
		\item Sample $\rPerm$ uniformly at random using public randomness. Identify the vertices of $\tau_i$ with $\rPerm(\pb_i)$ in $\pi$, and with a slight abuse of notation, any vertex $u$ of $\tau_i$ is used interchangeably with its counterpart in $\rPerm(\pb_i)$.
			In addition, each vertex $u$ of $\tau_i$ identifies its input given in $\tau_i$ with $\rB_i(u)$ in $\pi$.
		\item Do the following \emph{without} any communication:
			\begin{enumerate}
				\item Sample $\rM^{(1)}_{\sfP,i}$, conditioned on $\rPerm$, using public randomness.
				\item For each vertex $u$ of $\tau_i$, independently sample $\rT_i(u)$, conditioned on $\rB_i(u),\rM^{(1)}_{\sfP,i},\rPerm$, using \emph{private} randomness.
				\item Sample $\rM^{(1)}_{\sfP,-i}$, conditioned on $\rM^{(1)}_{\sfP,i},\rPerm$, using public randomness.
				\item Sample $\rM^{(1)}_\sfF$, conditioned on $\rM^{(1)}_\sfP,\rPerm$, using public randomness.
			\end{enumerate}
		\item For every $t \in [2,r]$, do the following with one round of communication:
			\begin{enumerate}
				\item For each vertex $u$ of $\tau_i$, generate and broadcast $\rM^{(t)}_{\sfP,i}(u)$ as in $\pi$, based on $\rG_i(u),\rM^{(< t)},\rPerm$.
				\item Sample $\rM^{(t)}_{\sfP,-i}$, conditioned on $\rM^{(< t)},\rM^{(t)}_{\sfP,i},\rPerm$, using public randomness.
				\item Sample $\rM^{(t)}_\sfF$, conditioned on $\rM^{(< t)},\rM^{(t)}_\sfP,\rPerm$, using public randomness.
			\end{enumerate}
		\item Let $\Gamma$ be the output of the referee of $\pi$ when given $\rM^{(\le r)}$. The referee of $\tau_i$ finally outputs $\Gamma \cap \paren{\rPerm(\pb_i) \times \rPerm(\pb_i)}$.
	\end{enumerate}
\end{Algorithm}

It is not hard to verify that~\Cref{lem:mis-sample-1,lem:mis-product-dist,lem:mis-sample-r} and~\Cref{cor:mis-tvd} also hold for approximate matching.
In fact, each of them follows verbatim as the proofs work in a black-box way.
Using all these results, we conclude this section with the proof of~\Cref{lem:apx-round-elim}.

\begin{proof}[Proof of~\Cref{lem:apx-round-elim}]
	For any $r$-round instance, let $\rO$ be the size of its maximum matching and for each $i \in [p_r]$, $\rO_i$ be the size of the maximum matching for the induced subgraph on $\rPerm(\pb_i)$.
	It always holds that
	\begin{align}
		\rO \ge \sum_{i \in [p_r]} \rO_i,\label{eqn:mm-lb}
	\end{align}
	since the union of maximum matchings for all principal $(r-1)$-round instances is always a valid matching for the $r$-round instance.
	Define $\rO^\pi$ to be the number of valid edges in $\Gamma \cap E$, where $E$ is the set of input edges of the $r$-round instance, and for each $i \in [p_r]$, $\rO^\pi_i$ to be the number of valid edges in $\Gamma \cap E \cap \paren{\rPerm(\pb_i) \times \rPerm(\pb_i)}$.
	It holds that
	\begin{align}
		\rO^\pi \le n_{r-1} \cdot f_r + \sum_{i \in [p_r]} \rO^\pi_i,\label{eqn:mm-ub}
	\end{align}
	because of the fact that the number of disjoint edges incident to fooling vertices is bounded by the total number of fooling vertices.
	(Recall that the output $\Gamma$ of the referee of $\pi$ must be a set of disjoint edges.)
	Also define $\rO^\tau_i$ to be the number of valid edges (i.e. excluding non-existing edges) output by the referee of $\tau_i$ for an $(r-1)$-round instance.
	Note that the referee of $\pi$ is a deterministic function of $\rM^{(\le r)}$ by assumption, so for each $i \in [p_r]$, the referee of $\tau_i$ is a deterministic function of $\rM^{(\le r)}$ and $\rPerm$, by~\Cref{alg:apx}.

	Again imagine the idealized situation where $\tau_i$ \emph{were} able to sample $\rM^{(\le r)},\rPerm \mid \rB_i$ precisely following $\mu_i$.
	By the linearity of expectation, we have
	\begin{align}
		\phantom{=} & \Exp_{i \in [p_r]}\Exp_{\rB_i \sim \mu_i}\Exp_{(\rM^{(\le r)},\rPerm \mid \rB_i) \sim \mu_i}\bracket{\rO^\tau_i}\nonumber\\
		= {} & \Exp_{i \in [p_r]}\Exp_{\rB_i \sim \mu}\Exp_{(\rM^{(\le r)},\rPerm \mid \rB_i) \sim \mu}\bracket{\rO^\tau_i} \tag{as $\mu_i(\rM^{(\le r)},\rB_i,\rPerm) = \mu(\rM^{(\le r)},\rB_i,\rPerm)$}\nonumber\\
		= {} & \Exp_{i \in [p_r]}\Exp_{(\rM^{(\le r)},\rG,\rPerm) \sim \mu}\bracket{\rO^\pi_i} \tag{by~\Cref{alg:apx}}\nonumber\\
		= {} & \Exp_{(\rG,\rPerm) \sim \mu}\Exp_{i \in [p_r]}\bracket{\rO^\pi_i} \tag{as $\rM^{(\le r)}$ is fully determined by $\rG,\rPerm$}\nonumber\\
		\ge {} & \Exp_{(\rG,\rPerm) \sim \mu}\bracket{\frac{\rO^\pi - n_{r-1} \cdot f_r}{p_r}} \tag{by~\Cref{eqn:mm-ub}}\nonumber\\
		= {} & \Exp_{(\rG,\rPerm) \sim \mu}\bracket{\frac{\rO^\pi}{p_r}} - \frac{n_{r-1} \cdot f_r}{p_r}\nonumber\\
		\ge {} & \alpha \cdot \Exp_{(\rG,\rPerm) \sim \mu}\bracket{\frac{\rO}{p_r}} - \frac{n_{r-1} \cdot f_r}{p_r} \tag{as $\pi$ has an approximation ratio of $\alpha^{-1}$}\nonumber\\
		\ge {} & \alpha \cdot \Exp_{(\rG,\rPerm) \sim \mu}\Exp_{i \in [p_r]}\bracket{\rO_i} - \frac{n_{r-1} \cdot f_r}{p_r} \tag{by~\Cref{eqn:mm-lb}}\nonumber\\
		= {} & \alpha \cdot \Exp_{i \in [p_r]}\Exp_{\rB_i \sim \nu_i}\bracket{\rO_i} - \frac{n_{r-1} \cdot f_r}{p_r}, \label{eqn:mm-alpha}
	\end{align}
	as $\mu(\rB_i) = \nu_i(\rB_i)$.
	Meanwhile, for each $i \in [p_r]$, \Cref{fact:tvd-small} bounds the real expected matching size of $\tau_i$ over $\nu_i$ as
	\begin{align}
		\phantom{=} & \Exp_{\rB_i \sim \nu_i}\Exp_{(\rM^{(\le r)},\rPerm \mid \rB_i) \sim \nu_i}\bracket{\rO^\tau_i}\nonumber\\
		= {} & \Exp_{\rB_i \sim \mu_i}\Exp_{(\rM^{(\le r)},\rPerm \mid \rB_i) \sim \nu_i}\bracket{\rO^\tau_i} \tag{as $\nu_i(\rB_i) = \mu_i(\rB_i)$}\nonumber\\
		\ge {} & \Exp_{\rB_i \sim \mu_i}\Exp_{(\rM^{(\le r)},\rPerm \mid \rB_i) \sim \mu_i}\bracket{\rO^\tau_i} - \frac{n_{r-1}}{2} \cdot \Exp_{\rB_i \sim \mu_i}\tvd{\mu_i(\rM^{(\le r)},\rPerm \mid \rB_i)}{\nu_i(\rM^{(\le r)},\rPerm \mid \rB_i)},\label{eqn:mm-diff}
	\end{align}
	since the size of any matching is at most half of the total number of vertices.
	Combining~\Cref{cor:mis-tvd} with~\Cref{eqn:mm-alpha,eqn:mm-diff}, we have
	\begin{align*}
		\phantom{=} & \Exp_{i \in [p_r]}\Exp_{\rB_i \sim \nu_i}\Exp_{(\rM^{(\le r)},\rPerm \mid \rB_i) \sim \nu_i}\bracket{\rO^\tau_i}\\
		\ge {} & \Exp_{i \in [p_r]}\Exp_{\rB_i \sim \mu_i}\Exp_{(\rM^{(\le r)},\rPerm \mid \rB_i) \sim \mu_i}\bracket{\rO^\tau_i} - \frac{n_{r-1}}{2} \cdot \Exp_{i \in [p_r]}\Exp_{\rB_i \sim \mu_i}\tvd{\mu_i(\rM^{(\le r)},\rPerm \mid \rB_i)}{\nu_i(\rM^{(\le r)},\rPerm \mid \rB_i)}\\
		\ge {} & \alpha \cdot \Exp_{i \in [p_r]}\Exp_{\rB_i \sim \nu_i}\bracket{\rO_i} - \frac{n_{r-1} \cdot f_r}{p_r} - \frac{n_{r-1}}{2} \cdot \frac{1}{k \cdot n_{r-1}}\\
		\ge {} & \paren{\alpha - \frac{n_{r-1} \cdot f_r}{p_r \cdot n_{r-1}/(4k)} - \frac{1/(2k)}{n_{r-1}/(4k)}} \cdot \Exp_{i \in [p_r]}\Exp_{\rB_i \sim \nu_i}\bracket{\rO_i},
	\end{align*}
	as $\rO_i \ge n_{r-1}/(4k)$ for each $i \in [p_r]$ by~\Cref{clm:apx-size}.
	Therefore, picking the best index $i^* \in [p_r]$ shows $\tau_{i^*}$ has an approximation ratio of at most
	\[
		\paren{\alpha - \frac{n_{r-1} \cdot f_r}{p_r \cdot n_{r-1}/(4k)} - \frac{1/(2k)}{n_{r-1}/(4k)}}^{-1} \le \paren{\alpha - O(\frac{1}{n_{r-1}})}^{-1},
	\]
	as claimed.
\end{proof}

\subsection{A Reduction to Bipartite Graphs}\label{sec:bipartite}

The following reduction, together with the lower bound for approximate matching for general graphs shown in~\Cref{sec:apx-proof}, concludes the proof of~\Cref{thm:mm}.

\begin{lemma}\label{lem:bipartite}
	For $r,\alpha \ge 1$, if there exists an $r$-round protocol for computing an $\alpha$-approximate bipartite matching for bipartite graphs, then there also exists an $r$-round protocol for computing a $2\alpha$-approximate matching for general graphs, with exactly the same bandwidth.
\end{lemma}

\begin{proof}
	Let $\pi$ be a protocol for bipartite graphs.
	We construct a corresponding protocol $\pi'$ for general graphs as follows.
	On a given input graph $G = (V,E)$ with $n$ vertices, the vertices in $\pi'$ jointly sample $z \in \set{0,1}^n$ uniformly at random using public randomness.
	Let $L = \set{v \in V \mid z_v = 0}$ and $R = \set{v \in V \mid z_v = 1}$.
	Note that all vertices agree on $L,R$ since they can be easily computed from $z$.
	Also let $G' = (V,E')$ be the bipartite subgraph of $G$ where only edges across the cut $(L,R)$ are preserved, i.e. $E' = \set{(u,v) \in E \mid z_u \ne z_v}$.
	Then all vertices run $\pi$ on $G'$ as if the neighbors of a vertex $v \in V$ are $N_{G'}(v) = \set{u \in N_G(v) \mid z_u \ne z_v}$.
	The referee of $\pi'$ simply outputs the answer given by the referee of $\pi$.

	It is easy to see $\pi'$ has the same number of rounds and exactly the same bandwidth as $\pi$.
	Moreover, observe that $G'$ is essentially a random bipartition of $G$ so half of the original edges are dropped in expectation.
	In particular, we have $\Exp_z[\mu(G')] \ge \mu(G) / 2$.
	(Recall that $\mu(\cdot)$ denotes the size of the maximum matching.)
	Therefore, an $\alpha$-approximate bipartite matching for $G'$ (over the randomness of $z$) is also a $2\alpha$-approximate matching for $G$ by definition.
\end{proof}

\subsection*{Acknowledgement} 

We are thankful to the anonymous reviewers of FOCS 2022 for helpful comments on the presentation of the paper. 

Sepehr Assadi is supported in part by a National Science Foundation CAREER award CCF-2047061, a gift from Google Research, and a Fulcrum award from Rutgers Research Council.
Gillat Kol is supported by a National Science Foundation CAREER award CCF-1750443 and by a BSF grant No. 2018325.

\clearpage

\bibliographystyle{halpha-abbrv}
\bibliography{general}

\newcommand{\etalchar}[1]{$^{#1}$}
\begin{thebibliography}{ANRW15}
\expandafter\ifx\csname url\endcsname\relax
  \def\url#1{\texttt{#1}}\fi
\expandafter\ifx\csname doi\endcsname\relax
  \def\doi#1{\burlalt{doi:#1}{http://dx.doi.org/#1}}\fi
\expandafter\ifx\csname urlprefix\endcsname\relax\def\urlprefix{URL }\fi
\expandafter\ifx\csname href\endcsname\relax
  \def\href#1#2{#2}\fi
\expandafter\ifx\csname burlalt\endcsname\relax
  \def\burlalt#1#2{\href{#2}{#1}}\fi

\bibitem[A17]{Assadi17ca}
S.~Assadi.
\newblock Combinatorial auctions do need modest interaction.
\newblock In {\em Proceedings of the 2017 {ACM} Conference on Economics and
  Computation, {EC} '17, Cambridge, MA, USA, June 26-30, 2017}, pages 145--162,
  2017.

\bibitem[A22]{Assadi22}
S.~Assadi.
\newblock A two-pass (conditional) lower bound for semi-streaming maximum
  matching.
\newblock In J.~S. Naor and N.~Buchbinder, editors, {\em Proceedings of the
  2022 {ACM-SIAM} Symposium on Discrete Algorithms, {SODA} 2022, Virtual
  Conference / Alexandria, VA, USA, January 9 - 12, 2022}, pages 708--742.
  {SIAM}, 2022.

\bibitem[ACG{\etalchar{+}}15]{AhnCGMW15}
K.~J. Ahn, G.~Cormode, S.~Guha, A.~McGregor, and A.~Wirth.
\newblock Correlation clustering in data streams.
\newblock In F.~R. Bach and D.~M. Blei, editors, {\em Proceedings of the 32nd
  International Conference on Machine Learning, {ICML} 2015, Lille, France,
  6-11 July 2015}, volume~37 of {\em {JMLR} Workshop and Conference
  Proceedings}, pages 2237--2246. JMLR.org, 2015.

\bibitem[ACK19]{AssadiCK19}
S.~Assadi, Y.~Chen, and S.~Khanna.
\newblock Polynomial pass lower bounds for graph streaming algorithms.
\newblock In {\em Proceedings of the 51st Annual {ACM} {SIGACT} Symposium on
  Theory of Computing, {STOC} 2019, Phoenix, AZ, USA, June 23-26, 2019.}, pages
  265--276, 2019.

\bibitem[AGM12a]{AhnGM12a}
K.~J. Ahn, S.~Guha, and A.~McGregor.
\newblock Analyzing graph structure via linear measurements.
\newblock In {\em Proceedings of the Twenty-third Annual ACM-SIAM Symposium on
  Discrete Algorithms}, SODA '12, pages 459--467. SIAM, 2012.
\newblock \urlprefix\url{http://dl.acm.org/citation.cfm?id=2095116.2095156}.

\bibitem[AGM12b]{AhnGM12b}
K.~J. Ahn, S.~Guha, and A.~McGregor.
\newblock Graph sketches: sparsification, spanners, and subgraphs.
\newblock In {\em Proceedings of the 31st {ACM} {SIGMOD-SIGACT-SIGART}
  Symposium on Principles of Database Systems, {PODS} 2012, Scottsdale, AZ,
  USA, May 20-24, 2012}, pages 5--14, 2012.
\newblock \doi{10.1145/2213556.2213560}.

\bibitem[AGM13]{AhnGM13}
K.~J. Ahn, S.~Guha, and A.~McGregor.
\newblock Spectral sparsification in dynamic graph streams.
\newblock In {\em Approximation, Randomization, and Combinatorial Optimization.
  Algorithms and Techniques - 16th International Workshop, {APPROX} 2013, and
  17th International Workshop, {RANDOM} 2013, Berkeley, CA, USA, August 21-23,
  2013. Proceedings}, pages 1--10, 2013.

\bibitem[AKLY16]{AssadiKLY16}
S.~Assadi, S.~Khanna, Y.~Li, and G.~Yaroslavtsev.
\newblock Maximum matchings in dynamic graph streams and the simultaneous
  communication model.
\newblock In {\em Proceedings of the Twenty-Seventh Annual {ACM-SIAM} Symposium
  on Discrete Algorithms, {SODA} 2016, Arlington, VA, USA, January 10-12,
  2016}, pages 1345--1364, 2016.

\bibitem[AKM22]{AssadiKM22}
S.~Assadi, P.~Kumar, and P.~Mittal.
\newblock Brooks' theorem in graph streams: {A} single-pass semi-streaming
  algorithm for {\(\Delta\)}-coloring.
\newblock {\em CoRR}, abs/2203.10984, 2022.

\bibitem[AKO20]{AssadiKO20}
S.~Assadi, G.~Kol, and R.~Oshman.
\newblock Lower bounds for distributed sketching of maximal matchings and
  maximal independent sets.
\newblock In Y.~Emek and C.~Cachin, editors, {\em {PODC} '20: {ACM} Symposium
  on Principles of Distributed Computing, Virtual Event, Italy, August 3-7,
  2020}, pages 79--88. {ACM}, 2020.

\bibitem[ANRW15]{AlonNRW15}
N.~Alon, N.~Nisan, R.~Raz, and O.~Weinstein.
\newblock Welfare maximization with limited interaction.
\newblock In V.~Guruswami, editor, {\em {IEEE} 56th Annual Symposium on
  Foundations of Computer Science, {FOCS} 2015, Berkeley, CA, USA, 17-20
  October, 2015}, pages 1499--1512. {IEEE} Computer Society, 2015.

\bibitem[AR20]{AssadiR20}
S.~Assadi and R.~Raz.
\newblock Near-quadratic lower bounds for two-pass graph streaming algorithms.
\newblock In S.~Irani, editor, {\em 61st {IEEE} Annual Symposium on Foundations
  of Computer Science, {FOCS} 2020, Durham, NC, USA, November 16-19, 2020},
  pages 342--353. {IEEE}, 2020.

\bibitem[BBH{\etalchar{+}}19]{BalliuBHORS19}
A.~Balliu, S.~Brandt, J.~Hirvonen, D.~Olivetti, M.~Rabie, and J.~Suomela.
\newblock Lower bounds for maximal matchings and maximal independent sets.
\newblock In D.~Zuckerman, editor, {\em 60th {IEEE} Annual Symposium on
  Foundations of Computer Science, {FOCS} 2019, Baltimore, Maryland, USA,
  November 9-12, 2019}, pages 481--497. {IEEE} Computer Society, 2019.

\bibitem[BHH19]{BehnezhadHH19}
S.~Behnezhad, M.~Hajiaghayi, and D.~G. Harris.
\newblock Exponentially faster massively parallel maximal matching.
\newblock In D.~Zuckerman, editor, {\em 60th {IEEE} Annual Symposium on
  Foundations of Computer Science, {FOCS} 2019, Baltimore, Maryland, USA,
  November 9-12, 2019}, pages 1637--1649. {IEEE} Computer Society, 2019.

\bibitem[BMN{\etalchar{+}}11]{BeckerMNRST11}
F.~Becker, M.~Matamala, N.~Nisse, I.~Rapaport, K.~Suchan, and I.~Todinca.
\newblock Adding a referee to an interconnection network: What can(not) be
  computed in one round.
\newblock In {\em 25th {IEEE} International Symposium on Parallel and
  Distributed Processing, {IPDPS} 2011, Anchorage, Alaska, USA, 16-20 May, 2011
  - Conference Proceedings}, pages 508--514. {IEEE}, 2011.

\bibitem[BMRT14]{BeckerMRT14}
F.~Becker, P.~Montealegre, I.~Rapaport, and I.~Todinca.
\newblock The simultaneous number-in-hand communication model for networks:
  Private coins, public coins and determinism.
\newblock In M.~M. Halld{\'{o}}rsson, editor, {\em Structural Information and
  Communication Complexity - 21st International Colloquium, {SIROCCO} 2014,
  Takayama, Japan, July 23-25, 2014. Proceedings}, volume 8576 of {\em Lecture
  Notes in Computer Science}, pages 83--95. Springer, 2014.

\bibitem[BMRT18]{BeckerMRT18}
F.~Becker, P.~Montealegre, I.~Rapaport, and I.~Todinca.
\newblock The impact of locality on the detection of cycles in the broadcast
  congested clique model.
\newblock In M.~A. Bender, M.~Farach{-}Colton, and M.~A. Mosteiro, editors,
  {\em {LATIN} 2018: Theoretical Informatics - 13th Latin American Symposium,
  Buenos Aires, Argentina, April 16-19, 2018, Proceedings}, volume 10807 of
  {\em Lecture Notes in Computer Science}, pages 134--145. Springer, 2018.

\bibitem[BO17]{BravermanO17}
M.~Braverman and R.~Oshman.
\newblock A rounds vs. communication tradeoff for multi-party set disjointness.
\newblock In {\em 58th {IEEE} Annual Symposium on Foundations of Computer
  Science, {FOCS} 2017, Berkeley, CA, USA, October 15-17, 2017}, pages
  144--155, 2017.

\bibitem[CDK19]{CormodeDK19}
G.~Cormode, J.~Dark, and C.~Konrad.
\newblock Independent sets in vertex-arrival streams.
\newblock In {\em 46th International Colloquium on Automata, Languages, and
  Programming, {ICALP} 2019, July 9-12, 2019, Patras, Greece}, pages
  45:1--45:14, 2019.

\bibitem[CKP{\etalchar{+}}21]{ChenKPS0Y21}
L.~Chen, G.~Kol, D.~Paramonov, R.~R. Saxena, Z.~Song, and H.~Yu.
\newblock Almost optimal super-constant-pass streaming lower bounds for
  reachability.
\newblock In S.~Khuller and V.~V. Williams, editors, {\em {STOC} '21: 53rd
  Annual {ACM} {SIGACT} Symposium on Theory of Computing, Virtual Event, Italy,
  June 21-25, 2021}, pages 570--583. {ACM}, 2021.

\bibitem[CT06]{CoverT06}
T.~M. Cover and J.~A. Thomas.
\newblock {\em Elements of information theory {(2.} ed.)}.
\newblock Wiley, 2006.

\bibitem[DK20]{DarkK20}
J.~Dark and C.~Konrad.
\newblock Optimal lower bounds for matching and vertex cover in dynamic graph
  streams.
\newblock In S.~Saraf, editor, {\em 35th Computational Complexity Conference,
  {CCC} 2020, July 28-31, 2020, Saarbr{\"{u}}cken, Germany (Virtual
  Conference)}, volume 169 of {\em LIPIcs}, pages 30:1--30:14. Schloss Dagstuhl
  - Leibniz-Zentrum f{\"{u}}r Informatik, 2020.

\bibitem[DKO14]{DruckerKO13}
A.~Drucker, F.~Kuhn, and R.~Oshman.
\newblock On the power of the congested clique model.
\newblock In M.~M. Halld{\'{o}}rsson and S.~Dolev, editors, {\em {ACM}
  Symposium on Principles of Distributed Computing, {PODC} '14, Paris, France,
  July 15-18, 2014}, pages 367--376. {ACM}, 2014.

\bibitem[DNO14]{DobzinskiNO14}
S.~Dobzinski, N.~Nisan, and S.~Oren.
\newblock Economic efficiency requires interaction.
\newblock In {\em Symposium on Theory of Computing, {STOC} 2014, New York, NY,
  USA, May 31 - June 03, 2014}, pages 233--242, 2014.

\bibitem[FKN21]{FiltserKN21}
A.~Filtser, M.~Kapralov, and N.~Nouri.
\newblock Graph spanners by sketching in dynamic streams and the simultaneous
  communication model.
\newblock In D.~Marx, editor, {\em Proceedings of the 2021 {ACM-SIAM} Symposium
  on Discrete Algorithms, {SODA} 2021, Virtual Conference, January 10 - 13,
  2021}, pages 1894--1913. {SIAM}, 2021.

\bibitem[GGK{\etalchar{+}}18]{GhaffariGKMR18}
M.~Ghaffari, T.~Gouleakis, C.~Konrad, S.~Mitrovic, and R.~Rubinfeld.
\newblock Improved massively parallel computation algorithms for mis, matching,
  and vertex cover.
\newblock In {\em Proceedings of the 2018 {ACM} Symposium on Principles of
  Distributed Computing, {PODC} 2018, July 23-27, 2018}, pages 129--138, 2018.

\bibitem[Gha16]{Ghaffari16}
M.~Ghaffari.
\newblock An improved distributed algorithm for maximal independent set.
\newblock In R.~Krauthgamer, editor, {\em Proceedings of the Twenty-Seventh
  Annual {ACM-SIAM} Symposium on Discrete Algorithms, {SODA} 2016, Arlington,
  VA, USA, January 10-12, 2016}, pages 270--277. {SIAM}, 2016.

\bibitem[GMT15]{GuhaMT15}
S.~Guha, A.~McGregor, and D.~Tench.
\newblock Vertex and hyperedge connectivity in dynamic graph streams.
\newblock In {\em Proceedings of the 34th {ACM} Symposium on Principles of
  Database Systems, {PODS} 2015, Melbourne, Victoria, Australia, May 31 - June
  4, 2015}, pages 241--247, 2015.

\bibitem[GO13]{GuruswamiO13}
V.~Guruswami and K.~Onak.
\newblock Superlinear lower bounds for multipass graph processing.
\newblock In {\em Proceedings of the 28th Conference on Computational
  Complexity, {CCC} 2013, K.lo Alto, California, USA, 5-7 June, 2013}, pages
  287--298, 2013.

\bibitem[GS62]{gale1962college}
D.~Gale and L.~S. Shapley.
\newblock College admissions and the stability of marriage.
\newblock {\em The American Mathematical Monthly}, 69(1):9--15, 1962.

\bibitem[JLN18a]{JurdzinskiL018a}
T.~Jurdzinski, K.~Lorys, and K.~Nowicki.
\newblock Communication complexity in vertex partition whiteboard model.
\newblock In Z.~Lotker and B.~Patt{-}Shamir, editors, {\em Structural
  Information and Communication Complexity - 25th International Colloquium,
  {SIROCCO} 2018, Ma'ale HaHamisha, Israel, June 18-21, 2018, Revised Selected
  Papers}, volume 11085 of {\em Lecture Notes in Computer Science}, pages
  264--279. Springer, 2018.

\bibitem[JLN18b]{JurdzinskiL018}
T.~Jurdzinski, K.~Lorys, and K.~Nowicki.
\newblock Communication complexity in vertex partition whiteboard model.
\newblock In Z.~Lotker and B.~Patt{-}Shamir, editors, {\em Structural
  Information and Communication Complexity - 25th International Colloquium,
  {SIROCCO} 2018, Ma'ale HaHamisha, Israel, June 18-21, 2018, Revised Selected
  Papers}, volume 11085 of {\em Lecture Notes in Computer Science}, pages
  264--279. Springer, 2018.

\bibitem[JN18]{Jurdzinski018b}
T.~Jurdzinski and K.~Nowicki.
\newblock Connectivity and minimum cut approximation in the broadcast congested
  clique.
\newblock In Z.~Lotker and B.~Patt{-}Shamir, editors, {\em Structural
  Information and Communication Complexity - 25th International Colloquium,
  {SIROCCO} 2018, Ma'ale HaHamisha, Israel, June 18-21, 2018, Revised Selected
  Papers}, volume 11085 of {\em Lecture Notes in Computer Science}, pages
  331--344. Springer, 2018.

\bibitem[KLM{\etalchar{+}}14]{KapralovLMMS14}
M.~Kapralov, Y.~T. Lee, C.~Musco, C.~Musco, and A.~Sidford.
\newblock Single pass spectral sparsification in dynamic streams.
\newblock In {\em 55th {IEEE} Annual Symposium on Foundations of Computer
  Science, {FOCS} 2014, Philadelphia, PA, USA, October 18-21, 2014}, pages
  561--570, 2014.
\newblock \doi{10.1109/FOCS.2014.66}.

\bibitem[KMW16]{KuhnMW16}
F.~Kuhn, T.~Moscibroda, and R.~Wattenhofer.
\newblock Local computation: Lower and upper bounds.
\newblock {\em J. {ACM}}, 63(2):17:1--17:44, 2016.

\bibitem[KW14]{KapralovW14}
M.~Kapralov and D.~P. Woodruff.
\newblock Spanners and sparsifiers in dynamic streams.
\newblock In {\em {ACM} Symposium on Principles of Distributed Computing,
  {PODC} '14, Paris, France, July 15-18, 2014}, pages 272--281, 2014.

\bibitem[Lin87]{Linial87}
N.~Linial.
\newblock Distributive graph algorithms-global solutions from local data.
\newblock In {\em 28th Annual Symposium on Foundations of Computer Science, Los
  Angeles, California, USA, 27-29 October 1987}, pages 331--335. {IEEE}
  Computer Society, 1987.

\bibitem[LMSV11]{LattanziMSV11}
S.~Lattanzi, B.~Moseley, S.~Suri, and S.~Vassilvitskii.
\newblock Filtering: a method for solving graph problems in mapreduce.
\newblock In {\em {SPAA} 2011: Proceedings of the 23rd Annual {ACM} Symposium
  on Parallelism in Algorithms and Architectures, San Jose, CA, USA, June 4-6,
  2011 (Co-located with {FCRC} 2011)}, pages 85--94, 2011.
\newblock \doi{10.1145/1989493.1989505}.

\bibitem[Lub85]{Luby85}
M.~Luby.
\newblock A simple parallel algorithm for the maximal independent set problem.
\newblock In R.~Sedgewick, editor, {\em Proceedings of the 17th Annual {ACM}
  Symposium on Theory of Computing, May 6-8, 1985, Providence, Rhode Island,
  {USA}}, pages 1--10. {ACM}, 1985.

\bibitem[MS15]{MaggsS15}
B.~M. Maggs and R.~K. Sitaraman.
\newblock Algorithmic nuggets in content delivery.
\newblock {\em Comput. Commun. Rev.}, 45(3):52--66, 2015.

\bibitem[MTVV15]{McGregorTVV15}
A.~McGregor, D.~Tench, S.~Vorotnikova, and H.~T. Vu.
\newblock Densest subgraph in dynamic graph streams.
\newblock In {\em Mathematical Foundations of Computer Science 2015 - 40th
  International Symposium, {MFCS} 2015, Milan, Italy, August 24-28, 2015,
  Proceedings, Part {II}}, pages 472--482, 2015.

\bibitem[Nis21]{Nisan21}
N.~Nisan.
\newblock The demand query model for bipartite matching.
\newblock In D.~Marx, editor, {\em Proceedings of the 2021 {ACM-SIAM} Symposium
  on Discrete Algorithms, {SODA} 2021, Virtual Conference, January 10 - 13,
  2021}, pages 592--599. {SIAM}, 2021.

\bibitem[NY19]{NelsonY19}
J.~Nelson and H.~Yu.
\newblock Optimal lower bounds for distributed and streaming spanning forest
  computation.
\newblock In T.~M. Chan, editor, {\em Proceedings of the Thirtieth Annual
  {ACM-SIAM} Symposium on Discrete Algorithms, {SODA} 2019, San Diego,
  California, USA, January 6-9, 2019}, pages 1844--1860. {SIAM}, 2019.

\bibitem[RS92]{roth1992two}
A.~E. Roth and M.~Sotomayor.
\newblock Two-sided matching.
\newblock {\em Handbook of game theory with economic applications}, 1:485--541,
  1992.

\bibitem[Yao77]{Yao77}
A.~C. Yao.
\newblock Probabilistic computations: Toward a unified measure of complexity
  (extended abstract).
\newblock In {\em 18th Annual Symposium on Foundations of Computer Science,
  Providence, Rhode Island, USA, 31 October - 1 November 1977}, pages 222--227,
  1977.

\bibitem[Yu21]{Yu21}
H.~Yu.
\newblock Tight distributed sketching lower bound for connectivity.
\newblock In D.~Marx, editor, {\em Proceedings of the 2021 {ACM-SIAM} Symposium
  on Discrete Algorithms, {SODA} 2021, Virtual Conference, January 10 - 13,
  2021}, pages 1856--1873. {SIAM}, 2021.

\end{thebibliography}

\clearpage
\appendix
\part*{Appendix}

\section{Basic Tools From Information Theory}\label{sec:info}

We now briefly introduce some definitions from information theory that are needed in this paper.
For a random variable $\rA$, we use $\supp{\rA}$ to denote the support of $\rA$ and $\distribution{\rA}$ to denote its distribution.
When it is clear from context, we may abuse the notation and use $\rA$ directly instead of $\distribution{\rA}$, for example, write
$A \sim \rA$ to mean $A \sim \distribution{\rA}$, i.e., $A$ is sampled from the distribution of random variable $\rA$.

We denote the \emph{Shannon entropy} of a random variable $\rA$ by
$\en{\rA}$, which is defined as:
\begin{align*}
	\en{\rA} = \sum_{A \in \supp{\rA}} \Pr\paren{\rA = A} \cdot \log\frac{1}{\Pr\paren{\rA = A}}.
\end{align*}
\noindent
The \emph{conditional entropy} of $\rA$ conditioned on $\rB$ is denoted by $\en{\rA \mid \rB}$ and defined as:
\begin{align*}
\en{\rA \mid \rB} = \Ex_{B \sim \rB} \bracket{\en{\rA \mid \rB = B}},
\end{align*}
where $\en{\rA \mid \rB = B}$ is defined in a standard way by using the distribution of $\rA$ conditioned on the event $\rB = B$ in the previous equation.
\noindent
The \emph{mutual information} of two random variables $\rA$ and $\rB$ is denoted by
$\mi{\rA}{\rB}$ and defined as:
\begin{align*}
\mi{\rA}{\rB} = \en{\rA} - \en{\rA \mid  \rB} = \en{\rB} - \en{\rB \mid  \rA}.
\end{align*}
\noindent
The \emph{conditional mutual information} $\mi{\rA}{\rB \mid \rC}$ is $\en{\rA \mid \rC} - \en{\rA \mid \rB,\rC}$ and hence by linearity of expectation:
\begin{align*}
	\mi{\rA}{\rB \mid \rC} = \Ex_{C \sim \rC} \bracket{\mi{\rA}{\rB \mid \rC = C}}.
\end{align*}

We also use the following standard measures of distance (or divergence) between distributions.

\paragraph{KL-divergence.} For two distributions $\mu$ and $\nu$, the \emph{Kullback-Leibler divergence} between $\mu$ and $\nu$ is denoted by $\kl{\mu}{\nu}$ and defined as:
\[
	\kl{\mu}{\nu} = \Ex_{a \sim \mu}\Bracket{\log\frac{\mu(a)}{\nu(a)}}.
\]

\paragraph{Total variation distance.} We denote the total variation distance between two distributions $\mu$ and $\nu$ on the same support $\Omega$ by $\tvd{\mu}{\nu}$, defined as:
\[
	\tvd{\mu}{\nu} = \max_{\Omega' \subseteq \Omega} \paren{\mu(\Omega')-\nu(\Omega')} = \frac{1}{2} \cdot \sum_{x \in \Omega} \card{\mu(x) - \nu(x)}.
\]

We refer the interested readers to the textbook by Cover and Thomas~\cite{CoverT06} for an excellent introduction to the field of information theory.

\subsection{Useful Properties of Entropy and Mutual Information}\label{sec:prop-en-mi}

We use the following basic properties of entropy and mutual information throughout.

\begin{fact}[cf.~\cite{CoverT06}]\label{fact:it-facts}
  Let $\rA$, $\rB$, $\rC$, and $\rD$ be four (possibly correlated) random variables.
   \begin{enumerate}
  \item \label{part:uniform} $0 \leq \en{\rA} \leq \log{\card{\supp{\rA}}}$. The right equality holds
    iff $\distribution{\rA}$ is uniform.
  \item \label{part:info-zero} $\mi{\rA}{\rB}[\rC] \geq 0$. The equality holds iff $\rA$ and
    $\rB$ are \emph{independent} conditioned on $\rC$.
  \item \label{part:cond-reduce} \emph{Conditioning on a random variable reduces entropy}:
    $\en{\rA \mid \rB,\rC} \leq \en{\rA \mid  \rB}$.  The equality holds iff $\rA \perp \rC \mid \rB$.
    \item \label{part:sub-additivity} \emph{Subadditivity of entropy}: $\en{\rA,\rB \mid \rC}
    \leq \en{\rA \mid C} + \en{\rB \mid  \rC}$.
   \item \label{part:ent-chain-rule} \emph{Chain rule for entropy}: $\en{\rA,\rB \mid \rC} = \en{\rA \mid \rC} + \en{\rB \mid \rC,\rA}$.
  \item \label{part:chain-rule} \emph{Chain rule for mutual information}: $\mi{\rA,\rB}{\rC \mid \rD} = \mi{\rA}{\rC \mid \rD} + \mi{\rB}{\rC \mid  \rA,\rD}$.
  \item \label{part:data-processing} \emph{Data processing inequality}: for a deterministic function $f(\rA)$, $\mi{f(\rA)}{\rB \mid \rC} \leq \mi{\rA}{\rB \mid \rC}$.
   \end{enumerate}
\end{fact}

\noindent
We also use the following propositions, regarding the effect of conditioning on mutual information.

\begin{proposition}\label{prop:info-increase}
  For random variables $\rA, \rB, \rC, \rD$, if $\rA \perp \rD \mid \rC$, then,
  \[\mi{\rA}{\rB \mid \rC} \leq \mi{\rA}{\rB \mid  \rC,  \rD}.\]
\end{proposition}
 \begin{proof}
  Since $\rA$ and $\rD$ are independent conditioned on $\rC$, by
  \itfacts{cond-reduce}, $\HH(\rA \mid  \rC) = \HH(\rA \mid \rC, \rD)$ and $\HH(\rA \mid  \rC, \rB) \ge \HH(\rA \mid  \rC, \rB, \rD)$.  We have,
	 \begin{align*}
	  \mi{\rA}{\rB \mid  \rC} &= \HH(\rA \mid \rC) - \HH(\rA \mid \rC, \rB) = \HH(\rA \mid  \rC, \rD) - \HH(\rA \mid \rC, \rB) \\
	  &\leq \HH(\rA \mid \rC, \rD) - \HH(\rA \mid \rC, \rB, \rD) = \mi{\rA}{\rB \mid \rC, \rD}. \qed
	\end{align*}

\end{proof}

\begin{proposition}\label{prop:info-decrease}
  For random variables $\rA, \rB, \rC,\rD$, if $ \rA \perp \rD \mid \rB,\rC$, then,
  \[\mi{\rA}{\rB \mid \rC} \geq \mi{\rA}{\rB \mid \rC, \rD}.\]
\end{proposition}
 \begin{proof}
 Since $\rA \perp \rD \mid \rB,\rC$, by \itfacts{cond-reduce}, $\HH(\rA \mid \rB,\rC) = \HH(\rA \mid \rB,\rC,\rD)$. Moreover, since conditioning can only reduce the entropy (again by \itfacts{cond-reduce}),
  \begin{align*}
 	\mi{\rA}{\rB \mid  \rC} &= \HH(\rA \mid \rC) - \HH(\rA \mid \rB,\rC) \geq \HH(\rA \mid \rD,\rC) - \HH(\rA \mid \rB,\rC) \\
	&= \HH(\rA \mid \rD,\rC) - \HH(\rA \mid \rB,\rC,\rD) = \mi{\rA}{\rB \mid \rC,\rD}.  \qed
 \end{align*}

\end{proof}

\subsection{Measures of Distance Between Distributions}\label{sec:prob-distance}

The following states the relation between mutual information and KL-divergence.

\begin{fact}\label{fact:kl-info}
	For random variables $\rA,\rB,\rC$,
	\[\mi{\rA}{\rB \mid \rC} = \Ex_{(b,c) \sim {(\rB,\rC)}}\Bracket{ \kl{\distribution{\rA \mid \rB=b,\rC=c}}{\distribution{\rA \mid \rC=c}}}.\]
\end{fact}

\noindent
We use the following basic properties of total variation distance.

\begin{fact}\label{fact:tvd-small}
	Suppose $\mu$ and $\nu$ are two distributions for a random variable $\rX$, then,
	\[
	\Ex_{\mu}\bracket{\rX} \leq \Ex_{\nu}\bracket{\rX} + \tvd{\mu}{\nu} \cdot \max_{X_0 \in \supp{\rX}} X_0.
	\]
\end{fact}

\begin{fact}\label{fact:tvd-chain-rule}
	Suppose $\mu$ and $\nu$ are two distributions for the tuple $(\rX_1,\ldots,\rX_t)$, then,
	\[
		\tvd{\mu(\rX_1,\ldots,\rX_t)}{\nu(\rX_1,\ldots,\rX_t)} \leq \sum_{i=1}^{n} \Exp_{(X_1,\ldots,X_{i-1}) \sim \mu} \tvd{\mu(\rX_i \mid X_1,\ldots,X_{i-1})}{\nu(\rX_i \mid X_1,\ldots,X_{i-1})}.
	\]
\end{fact}

\begin{fact}\label{fact:tvd-marginal}
	Suppose $\mu$ and $\nu$ are two distributions for the pair $(\rX,\rY)$, then,
	\[
		\tvd{\mu(\rX)}{\nu(\rX)} \le \tvd{\mu(\rX,\rY)}{\nu(\rX,\rY)}.
	\]
\end{fact}

\noindent
The following Pinsker's inequality bounds the total variation distance between two distributions based on their KL-divergence.

\begin{fact}[Pinsker's inequality]\label{fact:pinskers}
	For any distributions $\mu$ and $\nu$,
	$
	\tvd{\mu}{\nu} \leq \sqrt{\frac{1}{2} \cdot \kl{\mu}{\nu}}.
	$
\end{fact}

\end{document}